\documentclass[runningheads]{llncs}

\renewcommand{\paragraph}[1]{\noindent {\bf #1}}

\usepackage{graphicx}
 \PassOptionsToPackage{names,dvipsnames}{xcolor}

\usepackage{microtype}
\usepackage{amsmath}
\usepackage{comment}

\usepackage{mathtools}
\usepackage{proof}
\usepackage{hyperref}
\usepackage{mdframed}
\mdfdefinestyle{alttight}{innertopmargin=0pt,innerleftmargin=2pt,innerrightmargin=2pt}

\usepackage{xspace}
\usepackage{enumerate}
\usepackage{mathpartir}
\usepackage{bussproofs}

\usepackage[ruled, noend]{algorithm2e}
\usepackage[noend]{algorithmic}

\usepackage[firstpage]{draftwatermark}

\usepackage{wrapfig}
\usepackage{colortbl}

\usepackage[T1]{fontenc}

\usepackage{todonotes}

\usepackage[skins,listings,breakable]{tcolorbox}

\usepackage{stmaryrd}

\usepackage{listings}

\usepackage{subcaption}

\usepackage{tikz}
\usetikzlibrary{automata, positioning, arrows}

\usepackage{float}
\restylefloat{table}

\usepackage{arydshln}
\usepackage{pgfplots}
\pgfplotsset{width=10cm,compat=1.9}

\usepackage{mdframed}
\mdfdefinestyle{alttight}{innertopmargin=0pt,innerleftmargin=2pt,innerrightmargin=2pt}

\usepackage{bm}

\usepackage{amsfonts}

\usepackage{color}

\usepackage{thmtools}
\usepackage{thm-restate}

\definecolor{codegreen}{rgb}{0,0.6,0}
\definecolor{codegray}{rgb}{0.5,0.5,0.5}
\definecolor{codepurple}{rgb}{0.58,0,0.82}
\definecolor{backcolour}{rgb}{0.98,0.98,0.96}

\definecolor{pblue}{rgb}{0.13,0.13,1}
\definecolor{pgreen}{rgb}{0,0.5,0}
\definecolor{pred}{rgb}{0.9,0,0}
\definecolor{pgrey}{rgb}{0.46,0.45,0.48}
\definecolor{porange}{rgb}{1.0, 0.55, 0.0}
\definecolor{ppurple}{rgb}{0.41, 0.16, 0.38}
\definecolor{plightblue}{rgb}{0.0,0.35,0.6}

\definecolor{plightblue2}{rgb}{0.51,0.93,1.0}

\definecolor{colorrevision}{rgb}{0.0,0.35,0.6}

\newcommand{\longversion}[2]{\ifx\ifm\undefined#1\else#2\fi}

\usepackage{proof}
\usepackage{bussproofs}

\usepackage{listings}

\makeatletter
\newcommand*{\rom}[1]{\expandafter\@slowromancap\romannumeral #1@}
\makeatother

\newcommand{\secref}[1]{\S\,\ref{#1}}
\newcommand{\defref}[1]{Def.~\ref{#1}}

\newcommand{\figref}[1]{Fig.~\ref{#1}}
\newcommand{\thmref}[1]{Thm.~\ref{#1}}

\newcommand{\exref}[1]{Ex.~\ref{#1}}

\newcommand{\appref}[1]{App.~\ref{#1}}

\newcommand{\algref}[1]{Alg.~\ref{#1}}
\newcommand{\lstref}[1]{Listing~\ref{#1}}

\newcommand{\lta}{\textsf{BFA}\xspace}
\newcommand{\lfa}{\textsf{BFA}\xspace}
\newcommand{\lfas}{\textsf{BFA}s\xspace}

\newcommand{\wtd}[1]{\widetilde{#1}}

\newcommand{\len}[1]{|#1|}
\newcommand{\map}[1]{\ensuremath{{\llbracket}#1{\rrbracket}}}

\newcommand{\dtransfer}{\mathsf{dtransfer}}
\newcommand*{\defeq}{\stackrel{\text{def}}{=}}

\newcommand{\accesspath}{\ensuremath{\mathcal{AP}}}

\newcommand{\revision}[1]{#1} 

\newcommand{\revisiontodo}[1]{}

\SetWatermarkText{\hspace*{3.7in}\raisebox{8.55in}{\href{https://eapls.org/pages/artifact_badges/}{\mbox{\includegraphics[width=1.4cm]{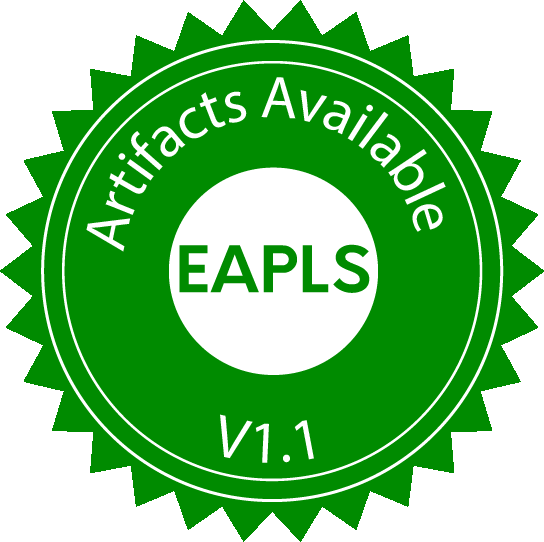}\includegraphics[width=1.4cm]{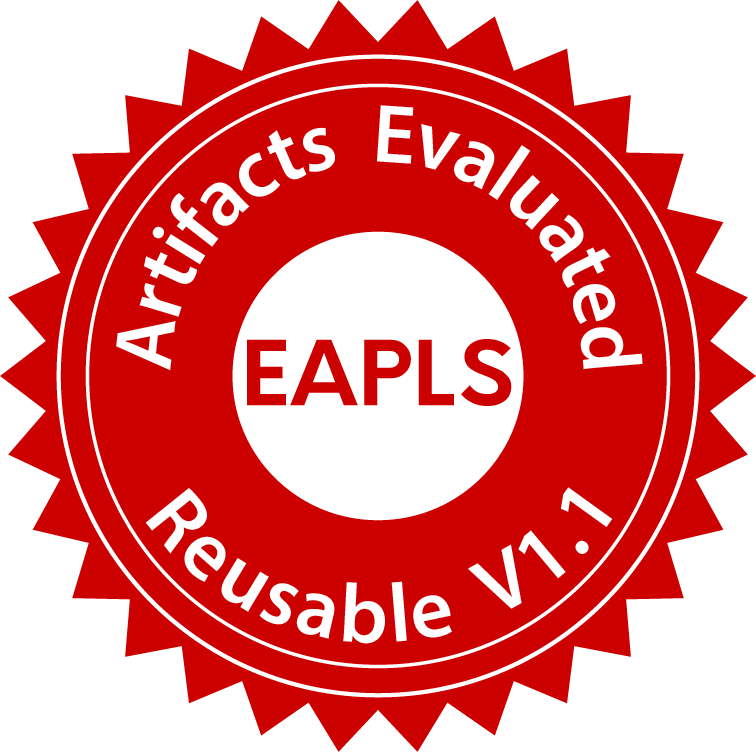}}}}}
\SetWatermarkAngle{0}

\begin{document}

\title{Scalable Typestate Analysis for Low-Latency Environments}

\author{Alen Arslanagi\'{c}\inst{1} \and 
Pavle Suboti\'{c}\inst{2} \and Jorge A. P\'{e}rez\inst{1}}

\institute{University of Groningen, The Netherlands \and Microsoft, Serbia 
}

\maketitle 

\begin{abstract}
Static analyses based on \emph{typestates} are important in certifying
correctness of code contracts. Such analyses rely on
Deterministic Finite Automata (DFAs) to specify properties of an object. 
We target the analysis of contracts in low-latency environments, where many useful contracts are impractical to codify as DFAs and/or the size of their associated DFAs leads
to sub-par performance. 
To address this bottleneck, we present a \emph{lightweight} typestate analyzer,
based on an expressive specification language that can succinctly specify code contracts. 
By implementing it in the 
static analyzer \textsc{Infer}, we demonstrate considerable performance and usability benefits when compared to existing
techniques. 
A central insight is to rely on a sub-class of DFAs with efficient \emph{bit-vector} operations.
\end{abstract}

\section{Introduction}
Industrial-scale software is generally composed of multiple interacting components, which are typically  
produced separately. As a result, software integration is a major source of 
bugs~\cite{integrationbugs}. Many integration bugs can be attributed to 
violations of \emph{code contracts}. Because these contracts are implicit and informal in nature,  
the resulting bugs are particularly insidious. To address this problem, formal 
code contracts are an effective solution~\cite{staticcontract},  
because static analyzers can automatically check whether client code adheres to ascribed contracts.

\emph{Typestate} is a fundamental concept in ensuring the correct use of
contracts and APIs. A typestate refines the concept of a type: whereas a type
denotes the valid operations on an object, a typestate denotes operations valid
on an object in its \emph{current program
context}~\cite{DBLP:journals/tse/StromY86}. Typestate analysis is a technique
used to enforce temporal code contracts. In object-oriented programs, where
objects change state over time, typestates denote the valid sequences of method
calls for a given object. The behavior of the object is prescribed by the
{collection of typestates}, and each method call can potentially change the
object's typestate. 

Given this, it is natural for static typestate checkers, such as
\textsc{Fugue}~\cite{DeLineECOOP2004}, \textsc{SAFE}~\cite{SAFE}, and
\textsc{Infer}'s \textsc{Topl} checker~\cite{topl}, to define the analysis property using
Deterministic Finite Automata (DFAs). The abstract domain of the analysis is a
set of states in the DFA; each operation on the object modifies the set of
possible reachable states. If the set of abstract states contains an error
state, then the analyzer  warns the user that a code contract may be violated.
Widely applicable and conceptually simple, DFAs are the de facto model in
typestate  analyses. 

Here we target the analysis of realistic code contracts in low-latency
environments such as, e.g., Integrated Development Environments
(IDEs)~\cite{inca,nblyzer}. In this context, to avoid noticeable disruptions in
the users' workflow, the analysis should ideally run \emph{under a
second}~\cite{rail}. However, relying on DFAs jeopardizes this goal, as it can
lead to scalability issues. Consider, e.g., a class with $n$ methods in which
each method \emph{enables} another one and then \emph{disables} itself: the
contract can lead to a DFA  with $2^{n}$ states. Even with a small $n$, such a
contract can be impractical to codify manually and will likely result in sub-par
performance.

Interestingly, many practical contracts do not require a full DFA. In our
enable/disable example, the method dependencies are \emph{local} to a subset
of methods: a enabling/disabling relation is established between pairs
of methods. DFA-based approaches have a \emph{whole class} expressivity; as a result, local method dependencies can impact transitions of unrelated methods. Thus, using DFAs for contracts that specify dependencies that are {local} to each
method (or to a few methods) is redundant and/or prone to inefficient
implementations. Based on this observation,  we present a \emph{lightweight} typestate analyzer for \emph{locally dependent} code contracts in low-latency 
environments. It rests upon two insights:
\begin{enumerate}
\item \emph{Allowed and disallowed sequences of method calls for objects can be succinctly specified without using DFAs.} 
To unburden the task of specifying typestates, we introduce  {\textit{lightweight annotations}} to specify
\emph{method dependencies} as annotations on methods. Lightweight annotations can specify 
code contracts for usage scenarios commonly encountered when using libraries such as File,
Stream, Socket, etc. in considerably fewer lines of code than DFAs. 
\item  \emph{A sub-class of DFAs suffices to express many useful code contracts}. 
To give semantics to lightweight annotations, we define \emph{Bit-Vector Finite Automata~(\lfas)}: a sub-class of DFAs whose analysis uses \emph{bit-vector} operations.
In many practical scenarios, \lfas suffice to capture information about the enabled and disabled methods at a given point. Because this information can be codified using bit-vectors, associated static analyses can be performed efficiently; in particular, our technique is not sensitive to the number of \lfa states, which in turn ensures scalability with contract and program size.
\end{enumerate}

\noindent
We have implemented our lightweight typestate analysis in the industrial-strength
static analyzer \textsc{Infer}~\cite{infer}. Our analysis exhibits concrete
usability and performance advantages and is expressive enough to encode
many relevant typestate properties in the literature. On
average, compared to state-of-the-art typestate analyses, our approach requires
less annotations than DFA-based analyzers and does not exhibit slow-downs due to
state increase. We summarise our contributions as follows:
\begin{itemize} 
	\item A specification language for typestates based on \emph{lightweight annotations} (\S\ref{sec:technical}). Our language rests upon \lfas, a new sub-class of DFA based on bit-vectors.
	\item A lightweight analysis technique for code contracts, implemented in \textsc{Infer} 
	(our artifact is available at \cite{lfachecker}).\footnote{Our code is available at
 \url{https://github.com/aalen9/lfa.git}}
        \item Extensive evaluations for our lightweight analysis technique,
        which demonstrate considerable gains in  performance and usability
        (\S\ref{sec:evaluation}). 
\end{itemize}
\noindent

\section{Bit-vector Typestate Analysis}
\label{sec:technical}
\subsection{Annotation Language}
\label{sec:annotations}
We introduce \lta specifications, which succinctly encode temporal properties by
only describing \textit{local method dependencies}, thus avoiding an explicit
DFA specification. \lfa specifications define code contracts by using atomic
combinations of annotations `$\texttt{@Enable}(n)$' and
`$\texttt{@Disable}(n)$', where $n$ is a set of method names. Intuitively,
`$\texttt{@Enable}(n) \ m$' asserts that invoking method $m$ makes calling
methods in $n$ valid in a continuation. Dually, `$\texttt{@Disable}(n) \ m$'
asserts that a call to $m$ disables calls to all methods in $n$ in the
continuation. More concretely, we give semantics for \lta annotations by
defining valid method sequences:

\begin{definition}[Annotation Language]
    Let $C = \{m_0, \ldots, m_n\}$ be a set of method names where each $m_i \in C$ 
    is annotated by  
    \begin{align*}
    	&\texttt{@Enable}(E_i) \ \texttt{@Disable}(D_i) \ m_i 
    \end{align*}
    \noindent where $E_i \subseteq C$, 
    $D_i \subseteq C$, and $E_i \cap D_i = \emptyset$. 
    Further, we have $E_0 \cup D_0 = C$. Let $s = x_0, x_1, x_2, \ldots$
    be a method sequence where each $x_i \in C$. 
    A sequence $s$ is \emph{valid (w.r.t. annotations)} if
    there is no substring $s'=x_i, \ldots,x_{k}$ of $s$ 
    such that $x_{k} \in D_i$ and $x_k \not\in E_j$, for 
     $j \in \{i+1, \ldots, k\}$.  
\end{definition}
The formal semantics for these specifications is given in \secref{sec:lfa}. We
note, if $E_i$ or $D_i$ is $\emptyset$ then we omit the corresponding
annotation.  Moreover, the \lfa language can be used to derive other useful
annotations defined as follows: 
\begin{align*}
    \texttt{@EnableOnly}(E_i) \ m_i  &\defeq \texttt{@Enable}(E_i) \ 
    \texttt{@Disable}(C \setminus E_i) \ m_i \\
    \texttt{@DisableOnly}(D_i) \ m_i &\defeq \texttt{@Disable}(D_i) \ 
    \texttt{@Enable}(C \setminus E_i) \ m_i \\
    \texttt{@EnableAll} \ m_i &\defeq \texttt{@Enable}(C) \ m_i
\end{align*} 
\noindent This way, `$\texttt{@EnableOnly}(E_i) \ m_i$' asserts that a call to method
$m_i$ enables only calls to methods in $E_i$ while disabling all other methods
in $C$; `$\texttt{@DisableOnly}(D_i) \ m_i$' is defined dually. 
Finally, `$\texttt{@EnableAll} \ m_i$' asserts that a call to method $m_i$  
enables all methods in a class;  `$\texttt{@DisableAll} \ m_i$' can be defined 
dually. 

To illustrate the expressivity and usability of \lta annotations, we consider the \texttt{SparseLU} class from \texttt{Eigen C++} 
library\footnote{\url{https://eigen.tuxfamily.org/dox/classEigen_1_1SparseLU.html}}.  
For brevity, we consider representative methods for a typestate specification
(we also omit return types): 
\lstset{
basicstyle=\scriptsize \ttfamily, 
}
\begin{lstlisting}
class SparseLU {
    void analyzePattern(Mat a); 
    void factorize(Mat a); 
    void compute(Mat a); 
    void solve(Mat b);  }
\end{lstlisting}

\noindent
The \texttt{SparseLU} class implements a lower-upper (LU) decomposition of a
sparse matrix. \texttt{Eigen}'s implementation uses assertions to dynamically
check that: (i) \texttt{analyzePattern} is called prior to \texttt{factorize}
and (ii) \texttt{factorize} or \texttt{compute} are called prior to
\texttt{solve}. At a high-level, this contract tells us that \texttt{compute}
(or \texttt{analyzePattern().factorize()}) prepares resources for invoking
\texttt{solve}.

We notice that there are method call sequences that do not cause errors, but
have redundant computations. For example,  we can disallow consecutive calls to
\texttt{compute} as in, e.g., sequences like
`\texttt{compute().compute().solve()}' as the result of the first compute is
never used. Further, \texttt{compute} is essentially implemented as
`\texttt{analyzePattern}().\texttt{factorize}()'. 
Thus, it is also redundant to call \texttt{factorize} after \texttt{compute}.
The  DFA that substitutes dynamic checks and avoids redundancies is given in
Figure~\ref{fig:sparselu-dfa-refined}. Following the
literature~\cite{DeLineECOOP2004}, this DFA can be annotated inside a class
definition as in \lstref{lst:refdfa}. Here states are listed in the class header
and transitions are specified by \textit{@Pre} and \textit{@Post} conditions on
methods. However, this specification is too low-level and unreasonable for
software engineers to annotate their APIs with, due to high annotation
overheads.

\lstset{
basicstyle=\scriptsize \ttfamily, 
}
\begin{figure} [t]
    \begin{subfigure}[b]{0.48\textwidth}
\begin{lstlisting}[caption={SparseLU DFA Contract}, label={lst:refdfa}, captionpos=b]
class SparseLU {
    $states q0, q1, q2, q3;$ 
    $@Pre(q0) @Post(q1)$
    $@Pre(q3) @Post(q1)$ 
    void analyzePattern(Mat a); 
    $@Pre(q1) @Post(q2)$
    $@Pre(q3) @Post(q2)$  
    void factorize(Mat a); 
    $@Pre(q0) @Post(q2)$ 
    $@Pre(q3) @Post(q2)$ 
    void compute(Mat a); 
    $@Pre(q2) @Post(q3)$
    $@Pre(q3)$ 
    void solve(Mat b);  }
\end{lstlisting}
\end{subfigure} 
\begin{subfigure}[b]{0.48\textwidth}
\begin{lstlisting}[caption={SparseLU \lfa Contract}, label={lst:reflta}, captionpos=b, 
    numbers=none]
class SparseLU {
    
    
    $@EnableOnly(factorize)$
    void analyzePattern(Mat a); 
    
    $@EnableOnly(solve)$
    void factorize(Mat a); 
    
    $@EnableOnly(solve)$
    void compute(Mat a); 

    $@EnableAll$
    void solve(Mat b);  }
\end{lstlisting}
\vspace*{\fill}
\end{subfigure}
\end{figure} 

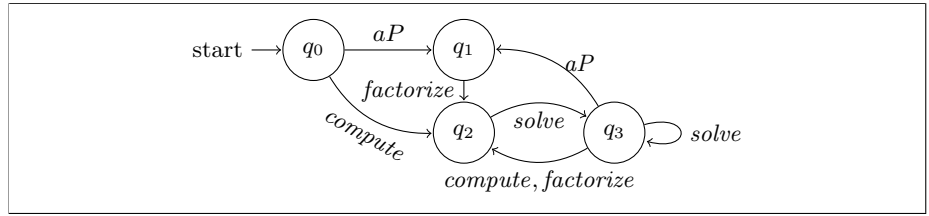
\begin{figure}[t!] 
    \begin{mdframed}
        \center
    \begin{tikzpicture}[shorten >=0.5pt,node distance=1.60cm,on grid,auto] 
       \node[state,initial] (q_0)   {$q_0$}; 
       \node[state](q_1) [right=2 cm of q_0] {$q_1$};
       \node[state](q_2) [below=1.1 cm of q_1] {$q_2$};
       \node[state](q_3) [right=2 cm of q_2] {$q_3$};
       \path[->]
       (q_0) edge node {$aP$} (q_1)
       (q_1) edge  [left] node {$\mathit{factorize}$} (q_2)
       (q_0) edge [bend right, sloped, below] node {\footnotesize \textit{compute}} (q_2)
       (q_2) edge [bend left, below] node {\footnotesize \textit{solve}} (q_3)
       (q_3) edge [bend left] node {$\mathit{compute}, \mathit{factorize}$} (q_2)
       (q_3) edge [bend right, right] node {$\mathit{aP}$} (q_1)
       (q_3) edge [loop right] node{$\mathit{solve}$} (q_3) 
       ;
\end{tikzpicture}
\end{mdframed}
\vspace*{-0.45cm}
\caption{SparseLU DFA}
\label{fig:sparselu-dfa-refined}
\end{figure}

In contrast, using \lfa annotations the entire \texttt{SparseLU} class contract
can be succinctly specified as  
in \lstref{lst:reflta}. Here, the starting state is unspecified; it is
determined by annotations. In fact,  methods that are not \textit{guarded} by
other methods (like \texttt{solve} is guarded by \texttt{compute}) are enabled
in the starting state. 
We remark that this can be overloaded by specifying
annotations on the constructor method. We can specify the contract with only 4
annotations; the corresponding DFA requires 8 annotations and 4 states specified
in the class header. We remark that a small change in local method dependencies by \lfa annotations
can result in a substantial change of the equivalent DFA. Let
$\{m_1,m_2,m_3,\ldots,m_n\}$ be methods of some class with DFA  associated (with
states $Q$) in which $m_1$ and $m_2$ are enabled in each state of $Q$. Adding
\texttt{@Enable(m2) m1} doubles the number of states of the DFA as we need the
set of states $Q$ where $m_2$ is enabled in each state, but also states from $Q$
with $m_2$ disabled in each state. Accordingly, transitions have to be
duplicated for the new states and the remaining methods ($m_3,\ldots,m_n$).

\subsection{Bit-vector Finite Automata}
\label{sec:lfa}

We define a class of DFAs, dubbed Bit-vector Finite Automata (\lfa), that
captures enabling/disabling dependencies between the methods of a class
leveraging a bit-vector abstraction on typestates.

\begin{definition}[Sets and Bit-vectors]
\label{d:bv}
Let $\mathcal{B}^n$ denote the set of bit-vec\-tors of length $n >0$. We write
$b, b', \ldots$ to denote elements of $\mathcal{B}^n$, with $b[i]$ denoting the
$i$-th bit in $b$. Given a finite set $S$ with $|S|=n$, every $A \subseteq S$
can be represented by a bit-vector $b_A \in \mathcal{B}^n$, obtained via the
usual characteristic function. By a small abuse of notation, given sets $A, A'
\subseteq S$, we may write $A \subseteq A'$ to denote the subset operation
applied on   $b_{A}$ and $b_{A'}$ (and similarly for $\cup,\cap$).
\end{definition}

\noindent
We first define a \lfa per class. Let us write $\mathcal{C}$ to denote the
finite set of all classes $c, c', \ldots$ under consideration. Given a $c \in
\mathcal{C}$ with $n$ methods, and assuming a total order on method names, we
represent them by the set $\Sigma_c = \{ m_1,\ldots,m_n \}$. 

A \lfa for a class with $n$ methods considers states $q_b$, where, following
Def.~\ref{d:bv}, the bit-vector  $b_A \in \mathcal{B}^n$ denotes the set $A \subseteq
\Sigma_c$ enabled at that point. We often write `$b$' (and $q_b$) rather than
`$b_A$' (and `$q_{b_A}$'), for simplicity. As we will see, the intent is that if
$m_i \in b$ (resp. $m_i \not\in b$), then the $i$-th method is enabled (resp.
disabled) in  $q_b$. Def.~\ref{d:map} will give a mapping from methods to
triples of bit-vectors. 
Given $k > 0$, let us write $1^k$ (resp. $0^k$) to denote a sequence of 1s
(resp. 0s) of length $k$. The initial state of the \lfa is then $q_{10^{n-1}}$, i.e., the
state in which only the first method is enabled and all the other $n-1$ methods
are disabled.

Given a class $c$, we define its associated mapping $\mathcal{L}_c$ as follows: 

\begin{definition}[Mapping $\mathcal{L}_c$]
\label{d:map}
Given a class $c$, we define $\mathcal{L}_c$ as a mapping from 
    methods to triples of subsets of $\Sigma_c$ as follows 
    $$\mathcal{L}_c : 
    \Sigma_c \to \mathcal{P}(\Sigma_c) \times 
    \mathcal{P}(\Sigma_c) \times  \mathcal{P}(\Sigma_c)$$
    \end{definition}
\noindent Given $m_i \in \Sigma_c$, we shall write $E_i$, $D_i$ and $P_i$ to denote each
of the elements of the triple $\mathcal{L}_c(m_i)$.  
The mapping $\mathcal{L}_c$ is induced by the annotations in class $c$:
for each   $m_i$, the sets  $E_i$ and $D_i$ are explicit, and $P_i$ is
simply the singleton $\{m_i\}$. 

In an \lfa, transitions between states $q_{b}, q_{b'}, \cdots$ are determined by
$\mathcal{L}_c$. Given $m_i \in \Sigma_c$, we have $j \in E_i$ if and only if the $m_i$
enables $m_j$; similarly, $k \in D_i$ if and only if $m_i$ disables $m_k$.
A transition from $q_b$ labeled by method $m_i$ leads to state $q_{b'}$, where
$b'$ is determined by $\mathcal{L}_c$ using $b$. Such a transition is defined
only if a pre-condition for $m_i$ is met in state $q_b$, i.e., 
 $P \subseteq b$. In that case, $b' = (b \cup E_i) \setminus D_i$.

These intuitions should suffice to illustrate our approach and, in particular,
the local nature of enabling and disabling dependencies between methods. The
following definition makes them precise.

\begin{definition}[\lfa]
    Given a  $c \in \mathcal{C}$ with $n > 0$ methods, a Bit-vector Finite
    Automaton (\lfa) for $c$ is defined as a tuple $M = (Q, \Sigma_c, \delta,
    q_{10^{n-1}},\mathcal{L}_c)$ where:
    \begin{itemize}
        \item $Q $ is a finite set of states $q_{10^{n-1}}, q_b, q_{b'}, \ldots$, where $b,
        b', \ldots \in \mathcal{B}^n$;
        \item $q_{10^{n-1}}$ is the initial state;
        \item $\Sigma_c =  \{m_1, \ldots, m_n\}$ is the alphabet (method
            identities);
        \item $\mathcal{L}_c$ is a \lta mapping (cf. Def.~\ref{d:map});

        \item $\delta: Q \times \Sigma_c \to Q$ is the transition function,
        where $\delta(q_b, m_i) = q_{b'}$ (with $b' = (b \cup E_i) \setminus D_i$)
        if $P_i \subseteq b$, and is
        undefined otherwise.
    \end{itemize}
    \label{d:lfa}
    \noindent\revision{We remark that in a \lfa  all states in $Q$ are accepting states.}
    \end{definition} 

    \begin{example}[SparseLU]
    We give the \lfa derived from the annotations in the SparseLU
    example (\lstref{lst:reflta}). We associate indices to methods:
    $$[0: \mathit{constructor},
    1:\mathit{aP},2:\mathit{compute},3:\mathit{factorize},4:\mathit{solve}]$$
    The constructor annotations are implicit: it
    enables methods that are not guarded by annotations on other methods (in
    this case, \textit{aP} and \textit{compute}). The mapping
    $\mathcal{L}_{\text{SparseLU}}$ is as follows: 
    \begin{align*}
        & \mathcal{L}_{\text{SparseLU}}  = \{ 0 \mapsto (\{1,2\},\{\},\{0\}),\ 
        1 \mapsto (\{3\},\{1,2,4\},\{1\}),  \\
        & \qquad 2 \mapsto (\{4\},\{1,2,3\},\{2\}),\  
        3 \mapsto (\{4\},\{1,2,3\},\{3\}),  4 \mapsto (\{1,2,3\},\{\},\{4\})\}
    \end{align*}
    The set of states is $Q = \{q_{1000},q_{1100}, q_{0010}, q_{0001},
    q_{1111}\}$ and the  transition function $\delta$  is given by following
    nine transitions: 
   
    \begin{center} 
    \begin{tabular}{lll}
        $\delta(q_{1000}, constr) = q_{1100}$ & 
        $\delta(q_{1100}, \mathit{aP}) = q_{0010}$ & 
        $\delta(q_{1100}, \mathit{compute}) = q_{0010}$
        \\ 
        $\delta(q_{0010}, \mathit{factorize}) = q_{0001}$  & 
        $\delta(q_{0001}, \mathit{solve}) = q_{1111}$  &
        $\delta(q_{1111}, \mathit{aP}) = q_{0010}$ 
        \\ 
        $\delta(q_{1111}, \mathit{compute}) = q_{0001}$ & 
        $\delta(q_{1111}, \mathit{factorize}) = q_{0001}$ &
        $\delta(q_{1111}, \mathit{solve}) = q_{1111}$
    \end{tabular} 
\end{center}
\end{example}

\paragraph{\lfas vs DFAs}
\revision{First, we need define some convenient notations:
\begin{definition}[Method sequences and concatenation]
    We use $\wtd m$ to denote a finite sequence of 
    method names in $\Sigma$. Further, 
    we use `$\cdot$' to denote sequence concatenation, defined as expected.
 
\end{definition}}
\revision{In the following theorem, we use $\hat \delta(q_b,\wtd m)$ to 
denote the extension of the one-step transition function $\delta(q_b, m_i)$ to a sequence of method calls 
({i.e.}, $\wtd m$).}
\lfas determine a strict sub-class of DFAs. 
\revision{First, because all states in $Q$ are accepting states, \lfa cannot encode the \emph{``must call''} property (cf. \secref{sec:relatedwork}).}
\revision{Next, we define the \emph{context-independency} property,
 satisfied by all \lfas but not by all DFAs:}
 \\
\revisiontodo{Define $\hat{\delta}$? Emphasize that all states in $Q$ are
accepting states, this is important for how we 
define $L$ in the theorem below. }
    \begin{restatable}[Context-independency]{thm}{thmbfa}
        \label{t:bfa}
            Let $M = (Q, \Sigma_c, \delta,
    q_{10^{n-1}},\mathcal{L}_c)$ be a \lfa. Also, 
    let 
    $L = \{ \wtd m : \hat{\delta}(q_{10^{n-1}}, \wtd m) = q' \wedge 
    q' \in Q \}$ be the language accepted by $M$. Then, for $m_n \in \Sigma_c$ we have  
    \begin{enumerate} 
        \item If there is $\wtd p \in L$ and $m_{n+1} \in \Sigma_c$ 
        s.t. $\wtd p \cdot m_{n+1} \notin L$ and 
        $\wtd p \cdot m_n \cdot m_{n+1} \in L$ then 
        there is no $\tilde m \in L$ s.t. 
        $\wtd m \cdot m_n \cdot m_{n+1} \notin L$. 

        \item If there is $\wtd p \in L$ and $m_{n+1} \in \Sigma_c$ 
        s.t. $\wtd p \cdot m_{n+1} \in L$ and 
        $\wtd p \cdot m_n \cdot m_{n+1} \notin L$ then 
        there is no $\widetilde m \in L$ s.t. 
        $\wtd m \cdot m_n \cdot m_{n+1} \in L$. 
    \end{enumerate} 
\end{restatable}

\begin{proof} 
    Directly by \defref{d:lfa}. 
    See~\longversion{\appref{app:tbfa}}{\cite{ASP22-full}} for details. 
\end{proof} 

 Informally, the above theorem tells that  previous calls ($\wtd m$) (\textit{i.e.},
 context) cannot impact the effect of a call to $m_n$ to subsequent calls
 ($m_{n+1}$). That is, Item 1. (resp. Item 2.) tells that method $m_n$ enables
 (resp. disables) the same set of methods in any context. For example, a DFA
 that disallows modifying a collection while iterating is not a \lfa (as in Fig.
 3 in~\cite{modulartypestate}). Let $it$ be a Java  Iterator with its usual
 methods for collection $c$.  For the illustration, we assume a single DFA
 relates the iterator and its collection methods. Then, the sequence
 `\texttt{it.hasNext;it.next;c.remove;it.hasNext}' should not be allowed, although
 `\texttt{c.remove;it.hasNext}' should be allowed. That is, \texttt{c.remove}
 disables \texttt{it.hasNext} \emph{only if} \texttt{it.hasNext} is previously
 called. Thus, the effect of calling \texttt{c.remove} depends on the calls that
 precedes it. 

\paragraph{\lfa subsumption}
Using \lfas, checking class subsumption boils down to usual set inclusion.
Suppose $M_1$ and $M_2$ are \lfas for classes $c_1$ and $c_2$, with $c_2$ being
the superclass of $c_1$. The class inheritance imposes an important question on
how we check that $c_1$ is a proper refinement of $c_2$. In other words, $c_1$
must subsume $c_2$: any valid sequence of calls to methods of $c_2$ must also be
valid for $c_1$. Using \lfas, we can verify this simply by checking annotations
method-wise.  We can check whether $M_2$ subsumes $M_1$  only by considering
their respective annotation mappings  $\mathcal{L}_{c_2}$ and
$\mathcal{L}_{c_1}$. 
Then, we have $M_2 \succeq M_1$
iff 
for all $m_j \in \mathcal{L}_{c_1}$ we have 
$E_1 \subseteq E_2$, $D_1 \supseteq D_2$, and $P_1 \subseteq P_2$
where $\langle E_i, D_i, P_i \rangle = \mathcal{L}_{c_i}(m_j)$  
for $i \in \{1,2\}$. 


\section{Compositional Analysis Algorithm}
\label{sec:algorithm}
Since \lfas can be ultimately encoded as bit-vectors, for 
the non-compositional case e.g., intra-procedural, standard data-flow 
analysis frameworks can be employed~\cite{DFA}. However, in the case of 
member objects methods being called, we present a compositional 
algorithm that is tailored for the \textsc{Infer} compositional static analysis 
framework. We motivate our compositional analysis technique with the example below.
\lstset{
basicstyle=\scriptsize \ttfamily, 
}
\begin{example} 
	\label{ex:ex-alg}
	Let \texttt{Foo} be a class 
	that has member \texttt{lu} of class \texttt{SparseLU} (cf. \lstref{lst:ex1}). 
\lstset{
basicstyle=\scriptsize \ttfamily, 
}
\begin{figure}[t]
    \begin{subfigure}[t]{0.48\textwidth}

\begin{lstlisting}[caption={Class Foo using SparseLU}, label={lst:ex1}, captionpos=b]
class Foo {
	SparseLU lu; Matrix a; 
	void setupLU1(Matrix b) {
		this.lu.compute(this.a); 
		if (?) this.lu.solve(b); }
	void setupLU2() {
		this.lu.analyzePattern(this.a); 
		this.lu.factorize(this.a); }
	void solve(Matrix b) {
		this.lu.solve(b); } }
\end{lstlisting}
\end{subfigure} 
\begin{subfigure}[t]{0.48\textwidth}
\begin{lstlisting}[caption={Client code for Foo}, label={lst:ex1-client}, captionpos=b, 
	numbers=none]
	void wrongUseFoo() { 
		Foo foo; Matrix b; 
		foo.setupLU1(); 
		foo.setupLU2(); 
		foo.solve(b); 
	}
\end{lstlisting}
\vspace*{\fill}
\end{subfigure}
 \vspace*{-0.28cm}
\end{figure} 
For each method of \texttt{Foo} that invokes methods on \texttt{lu} 
we compute a \emph{symbolic summary} that denotes the effect of executing that method on typestates of \texttt{lu}. 
To check against client code, a summary gives us: (i)~a pre-condition (i.e., which
methods should be allowed before calling a procedure) and (ii)~the effect on
the \emph{typestate} of an argument when returning from the procedure. 
A simple instance of a client is \texttt{wrongUseFoo} in \lstref{lst:ex1-client}. 

The central idea of our analysis is to accumulate enabling and
disabling annotations. For this, the abstract domain maps 
object access paths to triplets from
the definition of $\mathcal{L}_{\text{SparseLU}}$. A \emph{transfer function}
interprets method calls in this abstract state. We illustrate the
transfer function, presenting how abstract state evolves as
comments in the following code listing. 

\begin{lstlisting}
void setupLU1(Matrix b) {
	// s1 = this.lu -> ({}, {}, {}) 
	this.lu.compute(this.a); 
	// s2 = this.lu -> ({solve}, {aP, factorize, compute}, {compute})
	if (?) this.lu.solve(b); }
	// s3 = this.lu -> ({solve, aP, factorize, compute}, {}, {compute})
	// join s2 s3 = s4
	// s4 = sum1 = this.lu -> ({solve}, {aP, factorize, compute}, {compute})
\end{lstlisting}
\noindent At the procedure entry
(line 2) we initialize the abstract state as a triplet with empty sets ($s_1$).
Next, the abstract state is updated at the invocation of \texttt{compute}
(line 3): we copy the corresponding tuple from
$\mathcal{L}_{\text{SparseLU}}(compute)$ to obtain $s_2$ (line 4). 
Notice that \texttt{compute} is in the pre-condition set of $s_2$. 
Further, given the invocation of \texttt{solve} within the
if-branch in line 5 we transfer $s_2$ to $s_3$ as follows: the
enabling set of $s_3$ is the union of the enabling set from
$\mathcal{L}_{\text{SparseLU}}(solve)$ and the enabling set of $s_2$ with the disabling
set from $\mathcal{L}_{\text{SparseLU}}(solve)$ removed (i.e., an empty set
here).  
Dually, the disabling set of $s_3$ is the union of the disabling set of
$\mathcal{L}_{\text{SparseLU}}(solve)$ and the disabling set of $s_1$ with the enabling
set of $\mathcal{L}_{\text{SparseLU}}(solve)$ removed. Here we do not have
to add \texttt{solve} to the pre-condition set, as it is in the enabling set of $s_2$. 
Finally, we join the abstract states of two branches at
line 7 (i.e., $s_2$ and $s_3$). Intuitively, join operates as follows:
(i)~a method is enabled only if it is enabled in both branches and not disabled
in any branch; (ii)~a method is disabled if it is disabled in either branch;
(iii)~a method called in either branch must be in the pre-condition
(cf. \defref{d:lta-join}). Accordingly, in line 8 we obtain the final state
$s_4$ which is also a summary for \texttt{SetupLU1}.

\noindent
Now, we illustrate checking client code \texttt{wrongUseFoo()} 
with computed summaries: 

\begin{lstlisting}
void wrongUseFoo() { 
	Foo foo; Matrix b; 
	// d1 = foo.lu -> ({aP, compute}, {solve, factorize}, {})
	foo.setupLU1(); // apply sum1 to d1 
	// d2 = foo.lu -> ({solve}, {aP, factorize, compute}, {}) 
	foo.setupLU2(); // apply sum2 = {this.lu -> ({solve}, {aP, factorize, compute}, {aP}) } 
	// warning! `analyzePattern' is in pre of sum2, but not enabled in d2
	foo.solve(b); }
\end{lstlisting}

Above, at line 2 the abstract state is initialized with annotations of
constructor \texttt{Foo}. At the invocation of \texttt{setupLU1()} (line 4) we
apply $sum_1$ in the same way as user-entered annotations are applied to
transfer $s_2$ to $s_3$ above. Next, at line 6 we can see that \texttt{aP} is in
the pre-condition set in the summary for \texttt{setupLU2()} ($sum_2$), 
computed similarly as $sum_1$, but not in the enabling set
of the current abstract state $d_2$. Thus, a warning is raised:
\texttt{foo.lu} set up by \texttt{foo.setupLU1()} is never used and overridden by
\texttt{foo.setupLU2()}.

\paragraph{Class Composition}
In the above example, the allowed orderings of 
method calls to an object of class \texttt{Foo}  are imposed by
the contracts of its object members (\texttt{SparseLU}) and
the implementation of its methods. 
In practice, a class can have multiple members with their own
\lfa contracts. For instance, class \texttt{Bar} can use 
two solvers \texttt{SparseLU} and \texttt{SparseQR}: 
\begin{lstlisting}
class Bar { 
	SparseLU lu; SparseQR qr; /* ... */ }
\end{lstlisting}
\noindent where class \texttt{SparseQR} has its own \lfa contract. 
The implicit contract of \texttt{Bar} depends on 
contracts of both \texttt{lu} and \texttt{qr}. 
Moreover, a class as \texttt{Bar} can be a member of some other class. 
\revision{Thus, we refer to those classes as \emph{composed} and
to  classes that have declared contracts 
(as \texttt{SparseLU}) as {\emph{base classes}}. }

\end{example} 

\paragraph{Integrating Aliasing}
 Now, we discuss how \emph{aliasing information} can be integrated with our technique.
 In \exref{ex:ex-alg} member \texttt{lu} of object \texttt{foo} can be aliased.
 Thus, we keep track of \lfa triplets for all base members
 instead of  constructing an explicit \lfa contract for a composed class (e.g.,
 \texttt{Foo}). Further, we would need to
 generalize an abstract state to a mapping of \emph{alias sets} to \lfa
 triplets. That is, the elements of abstract state would be   
     $\{a_1,a_2,\ldots,a_n \} \mapsto \langle E, D, P \rangle$ 
  where $\{a_1,a_2,\ldots,a_n\}$ is a set of access paths. 
 For example, when invoking method \texttt{setupLU1}
 we would need to apply its summary ($sum_1$) to triplets of {each} alias set 
 that contains \texttt{foo.lu} as an element. 
 Let $d_1 = \{ S_1 \mapsto t_1, S_2 \mapsto t_2, \ldots \}$ be 
 an abstract state 
 where $S_1$ and $S_2$ are the only keys such that 
 $\texttt{foo.lu} \in S_i$ for $i \in \{1,2\}$ and $t_1$ and $t_2$ 
 are some \lfa triplets. 

 \begin{lstlisting}
 // d1 = S1 -> t1, S2 -> t2, ... 
 foo.setupLU1(); // apply sum1 = {this.lu -> t3}
 // d2 = S1 -> apply t3 to t1, S2 -> apply t3 to t2, ...
\end{lstlisting}
 
 Above, at line 2 we would need to update bindings of $S_1$ and $S_2$ (.resp) by
 applying an \lfa triplet for \texttt{this.foo} from $sum_1$, that is $t_3$, to
 $t_1$ and $t_2$ (.resp). 
 The resulting abstract state $d_2$ is given at line 4. 
 We remark that if a procedure does not alter aliases, 
 we can soundly compute and apply summaries, as shown above. 

\paragraph{Algorithm}
We formally define our analysis, which presupposes the control-flow graph (CFG)
of a program. Let us write \accesspath~to denote the set of access paths.
\revision{Access paths model heap locations as 
paths used to access them: a program variable 
followed by a finite sequence of field accesses (e.g., $foo.a.b$). 
We use access paths as we want to explicitly track  
states of class members.}
The abstract domain, denoted $\mathbb{D}$, maps access paths $\accesspath$ 
to \lfa triplets: 
\begin{align*}
	\mathbb{D}: \accesspath \to \bigcup_{c \in \mathcal{C}} Cod(\mathcal{L}_c)
\end{align*}
\noindent As variables denoted by an access path in $\accesspath$ can be of any
declared class $c \in \mathcal{C}$, 
the co-domain of $\mathbb{D}$ is the union of codomains of
$\mathcal{L}_c$ for all classes in a program. We remark that $\mathbb{D}$ is
sufficient for both checking and summary computation, as we will show in the
remaining of the section.

\begin{definition}[Join Operator]
\label{d:lta-join}
	We define $\bigsqcup: Cod(\mathcal{L}_c) \times Cod(\mathcal{L}_c) 
	\to Cod(\mathcal{L}_c)$
as follows: 
$\langle E_1, D_1, P_1 \rangle \sqcup 
\langle E_2, D_2, P_2 \rangle=  
\langle  E_1 \ \cap \ E_2 \setminus 
(\ D_1 \cup D_2),\ D_1 \cup D_2,\ P_1  \cup \ P_2 \rangle$.
\end{definition}
The join operator on $Cod(\mathcal{L}_c)$ is lifted to $\mathbb{D}$ by taking the union of
un-matched entries in the mapping. 


The compositional analysis is given in~\algref{alg:lta}. It
expects a program's CFG and a series of contracts, expressed as \lfas 
annotation mappings (\defref{d:map}). If the
program violates the \lfa contracts, a warning is raised. 
\revision{For the sake of clarity we only return a boolean
indicating if a contract is violated (cf. \defref{def:warning}). 
In the actual implementation we provide more
elaborate error reporting.} The algorithm
traverses the CFG nodes top-down. For each node $v$, it first
collects information from its predecessors (denoted by $\mathsf{pred}(v)$) and joins them
as $\sigma$ (line 3). 
Then,  the algorithm  checks whether a method can be called in the given
abstract state $\sigma$ by predicate \textsf{guard()} (cf.~\algref{alg:guard}).
If the pre-condition is met, then the \textsf{transfer()} function
(cf.~\algref{alg:transfer}) is called on a node. We assume a collection of \lfa
contracts (given as $\mathcal{L}_{c_1}, \ldots, \mathcal{L}_{c_k}$), which is
input for \algref{alg:lta}, is accessible in \algref{alg:transfer} to avoid
explicit passing. 
Now, we define some useful functions and predicates. 
For the algorithm, we require that the constructor
disabling set is the complement of the enabling set: 

\begin{definition}[${well\_formed}(\mathcal{L}_c)$]
	\label{d:wellformed-lc}
	Let $c$ be a class, $\Sigma$ methods set of class $c$, and 
	 $\mathcal{L}_c$. Then, $\mathsf{well\_formed}(\mathcal{L}_c)=\textbf{true}$  
	 iff $\mathcal{L}_c(constr) = \langle E,  \Sigma \setminus E, P \rangle$.
\end{definition}

\begin{definition}[${warning}(\cdot)$]
	\label{def:warning}
	Let $G$ be a CFG and 
		$\mathcal{L}_1,\ldots, \mathcal{L}_k$  be a 
		collection of \lfas. We define 
		$\mathsf{warning}(G, \mathcal{L}_1,\ldots, \mathcal{L}_k)= \textbf{true}$
		if there is a 
		path in $G$ that violates some of $\mathcal{L}_i$ for $i \in \{1, \ldots, k\}$.  
\end{definition}

\begin{definition}[$exit\_node(\cdot)$]
	Let $v$ be a method call node. Then,  $\textsf{exit\_node}(v)$ denotes 
	 exit node $w$ of a method body corresponding to $v$. 
\end{definition}
\revision{\begin{definition}[$actual\_arg(\cdot)$]
	\label{def:actual-arg}
	Let $v = Call-node[m_j(p_0:b_0, \ldots, p_n:b_n)]$ be a call node 
	where $p_0, \ldots, p_n$ are formal and $b_0, \ldots, b_n$ are actual arguments 
	and let $p \in \accesspath$. 
	We define $\textsf{actual\_arg}(p, v) = b_i$ if $p=p_i$ for $i \in \{0,\ldots,n \}$, otherwise 
	$\textsf{actual\_arg}(p, v)=p$. 
\end{definition}}

\revision{For convinience, we use  \emph{dot notation}  to access 
elements of \lfa triplets: 
	\begin{definition}[Dot notation for  \lfa triplets]
	Let $\sigma \in \mathbb{D}$ and $p \in \accesspath$. 
	Further, let 
	$\sigma[p] = \langle E_\sigma, D_\sigma, P_\sigma \rangle$.
	Then, we have $\sigma[p].E = E_\sigma$, $\sigma[p].D = D_\sigma$, and 
	$\sigma[p].P = P_\sigma$. 
\end{definition}}

\paragraph{Guard Predicate}
Predicate $\textsf{guard}(v, \sigma)$ checks whether a pre-condition for method
call node $v$ in the abstract state $\sigma$  is met (cf. \algref{alg:guard}). 
We represent a call node as $m_j(p_0:b_0,\ldots,p_n:b_n)$
where $p_i$ are formal and $b_i$ are actual arguments (for $i \in \{0, \ldots,
n\}$). Let $\sigma_w$ be a post-state of an exit node of method $m_j$. The
pre-condition is met if for all $b_i$ there are no elements in their
pre-condition set (i.e., the third element of $\sigma_w[b_i]$) that are also in
disabling set of the current abstract state $\sigma[b_i]$. 
For this predicate we need the property $D = \Sigma_{c_i}
\setminus E$, where $\Sigma_{c_i}$ is a set
of methods for class $c_i$.
This is ensured by condition 
$well\_formed(\mathcal{L}_{c_i})$ (\defref{d:wellformed-lc}) and 
by definition of \textsf{transfer()} (see below). 

	\LinesNumbered

	\begin{algorithm}[t]
		\caption{\lfa Compositional Analysis}
		\label{alg:lta}
		\SetKwProg{lta}{Procedure \emph{lta}}{}{end}
		\SetKwProg{transfer}{Procedure \emph{transfer}}{}{end}
		\KwData{{G} : A program's CFG, 
		a collection of \lfa mappings: 
		$\mathcal{L}_{c_1}, \ldots, \mathcal{L}_{c_k}$ over classes $c_1, \ldots c_k$
		such that $\mathit{well\_formed}(\mathcal{L}_{c_i})$ for $i \in \{1,\ldots,k \}$
		} 
		\KwResult{$warning(G, \mathcal{L}_{c_1}, \ldots, \mathcal{L}_{c_k})$}
		Initialize ${NodeMap} : Node \to \mathbb{D}$ as an empty map\; 
		\ForEach(){v in forward({G}))}
		{
			${\sigma} = \bigsqcup_{w \in pred(v)} w$\; 
			\leIf{{guard}(${v}$, ${\sigma}$)}
			{		
				${NodeMap[v]}$ := \emph{transfer}(${v}$,${\sigma}$);}
			{
			\Return{\textbf{True}}}
		}
	  \Return{\textbf{False}} 
		\end{algorithm}

		\begin{algorithm}[t]
			\caption{Guard Predicate}
			\label{alg:guard}
			\SetKwProg{lta}{Procedure \emph{lta}}{}{end}
			\SetKwProg{transfer}{Procedure \emph{transfer}}{}{end}
			\KwData{$v$ : CFG node, $\sigma$ : Domain}
			\KwResult{\textbf{False} iff $v$ is a method call that cannot be called 
			in $\sigma$}
			\SetKwProg{guard}{Procedure \emph{guard}}{}{end}
			\guard{$(v, \sigma)$}
				{
					\Switch{$v$}
					{
						\Case{Call-node[$m_j(p_0:b_0,\ldots,p_n:b_n)$]}
						{
							Let $w = exit\_node(v)$\; 
							\For{$i \in \{0, \ldots, n\}$}
							{
								\lIf{	
									$\sigma_w[p_i].P \cap \sigma[b_i].D \not= \emptyset$
									}{\Return{\textbf{False}}}	
							}
							\Return{\textbf{True}}
						}
						\Other{
						\Return{\textbf{True}}
						}
					}
				}
			\end{algorithm}

\begin{algorithm}[!t]
	\caption{Transfer Function}
	\label{alg:transfer}
	\SetKwProg{lta}{Procedure \emph{lta}}{}{end}
	\SetKwProg{transfer}{Procedure \emph{transfer}}{}{end}
	\KwData{
		$v$ : CFG node, $\sigma$ : Domain}
	\KwResult{Output abstract state $\sigma' : Domain$}
	\transfer{$(v, \sigma)$}{
			\Switch{$v$}
			{
				\Case{Entry-node[$m_j(p_0,\ldots,p_n)$]}
				{	
					Let $c_i$ be the class of method $m_j(p_0,
						\ldots, p_n)$\; 
					\leIf{There is $\mathcal{L}_{c_i}$}
					{	
						\Return{$\{ this \mapsto \mathcal{L}_{c_i}(m_j) \}$}\;
					}
					{
						\Return{EmptyMap}
					}
				}
				\Case{Call-node[$m_j(p_0:b_0,\ldots,p_n:b_n)$]}
				{
					Let  $\sigma_w$ be an abstract state of $exit\_node(v)$\; 
					Initialize $\sigma' := \sigma$\; 
					\If{\texttt{this} not in $\sigma'$}{
					
					\revision{\For{$ap$ in $dom(\sigma_w)$}{
						$ap' = {actual\_arg}(ap\{b_0 / \mathtt{this} \}, v)$\;
					\eIf{$ap'$ in $dom(\sigma)$}{
						$E' = (\sigma[ap'].E \ \cup \ \sigma_w[ap].E) 
						\setminus \sigma_w[ap].D$\;
						$D' = (\sigma[ap'].D \ \cup \ \sigma_w[ap].D) 
						\setminus \sigma_w[ap].E$\; 
						$P' = \sigma[ap'].P \ \cup \ (\sigma_w[ap].P \ 
						\setminus \ \sigma[ap'].E)$\;
						$\sigma'[ap']= \langle E', D', P' \rangle$\; 
					}{
						$\sigma'[ap']$ := $\sigma_w[ap]$\;
					}
					}
					}
					}
						\Return{$\sigma'$}
					}
					\Other{\Return{$\sigma$}}
				}
		}

		\end{algorithm}

	\paragraph{Transfer Function}
	The transfer function is given in~\algref{alg:transfer}. 
	It distinguishes between two types of CFG nodes: 

		\textbf{Entry-node:} (lines 3--6) This is a function entry node. For
		simplicity we represent it as $m_j(p_0, \ldots, p_n)$ where $m_j$ is a
		method name and $p_0, \ldots, p_n$ are formal arguments. 
		We assume $p_0$ is a reference to the receiver object
		(i.e., \textit{this}). 
		If method $m_j$ is defined in class
		$c_i$ that has user-supplied annotations $\mathcal{L}_{c_i}$, 
		in line 5 we initialize the domain to the singleton
		map (\textit{this} mapped to $\mathcal{L}_{c_i}(m_j)$). 
		Otherwise, we return an empty map meaning that a
		summary has to be computed. 
		
		\textbf{Call-node:} (lines 7--20)  
		We represent a call node as $m_j(p_0:b_0, \ldots, p_n:b_n)$ where  
		we assume actual arguments $b_0, \ldots, b_n$ are access paths for objects 
		and 
		$b_0$ represents a receiver object. 
		The analysis is skipped if \textit{this} is in the domain (line 10):
		this means the method has user-entered annotations.
		\revision{Otherwise, we transfer an abstract state for each argument 
		$b_i$, but also for each \emph{class member} whose state is updated by $m_j$.}
		\revision{Thus, we consider all access paths in the domain of $\sigma_w$, 
		that is $ap \in dom(\sigma_w)$ (line 11).
		We construct access path $ap'$ given $ap$.   
		 We distinguish two cases: $ap$ denotes (i)~a member and  
		 (ii)~a formal argument of $m_j$. By line 12 we handle both cases. 
		 In the former case we know $ap$ has form $this.c_1. \ldots . c_n$. 
		 We construct $ap'$ as $ap$ with ${this}$ substituted 
		 for $b_0$ ($\textsf{actual\_arg}(\cdot)$ is the identity 
		 in this case, see \defref{def:actual-arg}): e.g., if receiver $b_0$ is $this.a$ 
		 and $ap$ is $this.c_1. \ldots. c_n$ then $ap' = this.a.c_1.\ldots.c_n$. 
		 In the latter case $ap$ denotes formal argument $p_i$ and 
		$\textsf{actual\_arg}(\cdot)$ returns corresponding actual argument 
		   $b_i$ (as  $p_i\{b_0 / this\} = p_i$).  
		   Now, as $ap'$ is determined we construct its \lfa triplet.
		If $ap'$ is not in the domain of $\sigma$ (line 13) 
		we copy a corresponding \lfa triplet from $\sigma_w$ (line 19).  
		Otherwise, we transfer elements of an \lfa triplet at $\sigma[ap']$  as follows.} 
		The resulting enabling set is obtained by (i)~adding methods that $m_j$
		enables ($\sigma_w[ap].E$)  to the current enabling set $\sigma[ap'].E$, and 
		(ii)~removing
		methods that $m_j$ disables ($\sigma_w[ap].D$), from it. 
		The disabling set $D'$
		is constructed in a complementary way.  
		Finally, the pre-condition set $\sigma[ap'].P$ is expanded with elements of
		$\sigma_w[ap].P$ that are not in the enabling set $\sigma[ap'].E$. 
		We remark that the
		property $D = \Sigma_{c_i} \setminus E$ is preserved by the definition
		of $E'$ and $D'$. Transfer is the identity on $\sigma$ for all other
		types of CFG nodes. We can see that for each method call we have
		constant number of bit-vector operations per argument. That is, \lfa
		analysis is insensitive to the number of states, as a set of states is
		abstracted as a single set. 
		
		 Note, in our implementation we use several features specific to \textsc{Infer}:
(1)~\textsc{Infer}'s summaries which allow us to use a single  domain for intra
and inter procedural analysis; 
(2)~scheduling on CFG top-down traversal which
simplify the handling of branch statements. In principle, \lfa can be
implemented in other frameworks e.g., IFDS~\cite{IFDS}.

\paragraph{Correctness}
\label{s:correctness}
In a \lfa, we can abstract a set of states by the \emph{intersection}
of states in the set. 
That is, for $P \subseteq
Q$ all method call sequences accepted by each state in $P$ are also
accepted by the state that is the intersection of bits of states in the set. 
Theorem~\ref{t:state-intersection} formalizes this property. First we need an
auxiliary definition; let us write  $Cod(\cdot)$ to denote the codomain of a
mapping:

\begin{definition}[$\map{\cdot}(\cdot)$]
    \label{d:lfa-apply}
    Let $\langle E, D, P \rangle \in Cod(\mathcal{L}_c)$ and 
    $b \in \mathcal{B}^n$. 
	We define \\ $\map{\langle E, D, P \rangle}(b) = {b'}$ where 
	$b'=(b \cup E) \setminus D$ if $P \subseteq b$, and is  undefined otherwise. 
\end{definition}

\begin{restatable}[\lfa $\cap$-Property]{thm}{thmstates}
    \label{t:state-intersection}
	Let $M = (Q, \Sigma_c, \delta, q_{10^{n-1}}, \mathcal{L}_c)$,  
	$P \subseteq Q$, and $b_* = \bigcap_{q_b \in P} b$, then 
	\begin{enumerate}
		\item For $m \in \Sigma_c$,  it holds: 
        $\delta(q_b, m)$ is defined for all  $q_b \in P$  
		iff $\delta(q_{b_*}, m)$ is defined. 
		\item Let $\sigma=\mathcal{L}_c(m)$. 
		If $P' = \{ \delta(q_b, m) : q_b \in  P \}$  
		then $\bigcap_{q_b \in P'} b = \map{\sigma}(b_*)$. 
	\end{enumerate} 
\end{restatable}
\begin{proof} 
	By induction on cardinality of $P$ and \defref{d:lfa}. 
	See~\longversion{\appref{app:states}}{\cite{ASP22-full}} for details.
\end{proof}

Our \lfa-based algorithm (\algref{alg:lta}) interprets method
call sequences in the abstract state and joins them (using join from
\defref{d:lta-join}) following the control-flow of the program. Thus, we can
prove its correctness by separately establishing: (1) the correctness of the
interpretation of call sequences using 
a \emph{declarative} representation of the transfer function (\defref{d:transfer})  and
(2) the soundness of join operator
(\defref{d:lta-join}). 
For brevity, we consider a single program 
object, as method call sequences for distinct objects are
analyzed independently.
We define the \emph{declarative} transfer function as follows:

\begin{definition}[$\dtransfer_c(\cdot)$]
	\label{d:transfer}
	Let $c \in \mathcal{C}$ be a class, $\Sigma_c$ be a set of methods of $c$, and 
	$\mathcal{L}_c$ be a \lfa. 
	Further, let $m \in \Sigma_c$ be a method,  
	$\langle E^m, D^m, P^m \rangle =\mathcal{L}_c(m)$, and 
	$\langle E, D, P \rangle \in Cod(\mathcal{L}_c)$. Then, 
\begin{align*} 
	\dtransfer_{c}(m, \langle E, D, P \rangle) = 
	\langle E', D', P' \rangle
\end{align*} 

	\noindent where $E' = (E \ \cup \ E^{m}) \setminus D^{m}$, 
	$D' = (D \ \cup \ D^{m}) \setminus E^{m}$, 
	and $P' = P \ \cup \ (P^{m} \ \setminus \ E)$, 
	if $P^m \cap D = \emptyset$, and is undefined otherwise. 
	Let $m_1,\ldots,m_n, m_{n+1}$ be a method sequence and 
	$\phi = \langle E, D, P \rangle$, then 
	\begin{align*} 
		&\dtransfer_c(m_1,\ldots,m_n, m_{n+1}, \phi) = 
\dtransfer_c(m_{n+1}, \dtransfer_c(m_1,\ldots,m_n, \phi))
	\end{align*} 
\end{definition}

\noindent Relying on \thmref{t:state-intersection}, we 
state the soundness of $\mathsf{join}$: 

 \begin{restatable}[Soundness of $\sqcup$]{thm}{thmjoin}
	\label{t:join}
		Let $q_b \in Q$ and 
	$\phi_i = \langle E_i, D_i, P_i \rangle$ for $i \in \{1,2\}$. 
	Then, 	$\map{\phi_1}(b) \cap \map{\phi_2}(b) = \map{\phi_1 \sqcup \phi_2}(b)$. 
\end{restatable}

\begin{proof} 
	By definitions \defref{d:lta-join} and \defref{d:lfa-apply}, and set laws. 
	See~\longversion{\appref{app:join}}{\cite{ASP22-full}} for details. 
\end{proof} 

With these auxiliary notions in place, we show the correctness 
of the transfer function (i.e., summary computation that is  
specialized for the code checking):  
\begin{restatable}[Correctness of $\dtransfer_{c}(\cdot)$]{thm}{thmdtranfer}
	\label{t:dtranfer}
	Let $M = (Q, \Sigma, \delta, q_{10^{n-1}}, \mathcal{L}_c)$. 
	Let $q_b \in Q$ and $m_1 \ldots m_n \in \Sigma^*$. Then 
	\begin{align*}
			&\dtransfer_{c}(m_1 \ldots m_n, 
			\langle \emptyset, \emptyset, \emptyset, \rangle) =
			\langle E', D', P' \rangle
			\iff \hat{\delta}(q_b, m_1 \ldots m_n)=q_{b'} 
	\end{align*}
	\noindent where
	$b' = \map{\langle E', D', P' \rangle}(b)$. 
\end{restatable}

\begin{proof} 
	By induction on the length of the method call sequence. 
	See~\longversion{\appref{app:transfer}}{\cite{ASP22-full}} for details. 
\end{proof}

\section{Evaluation}
\label{sec:evaluation}
We evaluate our technique to validate the following two claims:
\begin{description}\itemsep-1pt
	\item[\textbf{\em Claim-I: Smaller annotation overhead.}] The \lfa contract annotation overheads 
        are smaller in terms of atomic annotations (e.g., @Post(...), @Enable(...)) than both competing analyses. 
	\item[\textbf{\em Claim-II: Improved scalability on large code and contracts.}] Our analysis scales 
        better than the competing analyzers for our use case on two dimensions, namely, caller code size and contract size.
\end{description}

\paragraph{Experimental Setup}
We used an Intel(R) Core(TM) i9-9880H  CPU  at  2.3
GHz  with  16GB  of  physical  
RAM running macOS 11.6  on the bare-metal. The experiments were conducted in
isolation without virtualization so that runtime results are robust.  All
experiments shown here are run in single-thread for \textsc{Infer} 1.1.0 running
with OCaml 4.11.1.

We implement two analyses in \textsc{Infer}, namely \lfa and DFA, and use the 
default \textsc{Infer} typestate analysis \textsc{Topl} as a baseline comparison. More in details:
(1) \lfa: The \textsc{Infer} implementation of the technique described in this
paper. (2) DFA: A lightweight DFA-based typestate implementation based on an DFA-based analysis 
implemented in \textsc{Infer}. We translate \lfa annotations to a
minimal DFA and perform the analysis.  (3) \textsc{Topl}: An industrial typestate analyzer, implemented in
\textsc{Infer}~\cite{topl}. This typestate analysis is designed for high
precision and not for low-latency environments. It uses \textsc{Pulse}, an \textsc{Infer}
memory safety analysis, which provides it with alias information. We include it
in our evaluation as a baseline state-of-the-art typestate analysis, i.e., an
off-the-shelf industrial strength tool we could hypothetically use. We note our 
benchmarks do not require aliasing and in theory \textsc{Pulse} is not required.

We analyze a benchmark of 18 contracts that specify common patterns of locally
dependent contract annotations for a class. 
\revision{Moreover, we auto-generate 122 client programs 
parametrized by lines of code, number of composed classes,  
if-branches, and loops. Note, the code is such that it
does not invoke the need for aliasing (as we do not support it yet in our \lfa
implementation).}
Client programs follow the compositional patterns we described
in~\exref{ex:ex-alg}; which can also be found in~\cite{papaya}. 

The annotations
for \lfa are manually specified; from them, we generate minimal DFAs
representations in DFA annotation format and \textsc{Topl} annotation format.

Our use case is to integrate static analyses in interactive IDEs e.g., Microsoft
Visual Studio Code~\cite{nblyzer}, so that code can
be analyzed at coding time. For this reason, our use case requires low-latency
execution of the static analysis. Our SLA is based on the RAIL user-centric
performance model~\cite{rail}.

\paragraph{Usability Evaluation} Fig.~\ref{fig:specsize} outlines the key features of the 18 contracts we
considered, called CR-1 -- CR-18. In~\longversion{\appref{app:cr4}}{\cite{ASP22-full}}
 we detail CR-4 as an example. For each contract,
we specify the number of methods, the number of DFA states the contract corresponds to, and number of 
atomic annotation terms in \lfa, DFA, and \textsc{Topl}. An atomic annotation term is a standalone annotation in the 
given annotation language.  We can observe that as the contract sizes
increase in number of states, the annotation overhead for DFA and \textsc{Topl} increase
significantly. On the other hand, the annotation overhead for \lfa remain largely
constant wrt. state increase and increases rather proportionally with the number
of methods in a contract. Observe that for contracts on classes with 4 or more
methods, a manual specification using DFA or \textsc{Topl} annotations becomes
impractical. Overall, we validate Claim-I by the fact that \lfa requires less
annotation overhead on all of the contracts, making contract specification more 
practical.

\begin{figure}[!t]
\begin{center}
    \resizebox{12cm}{!}{
    
\footnotesize
\begin{tabular}{l l} 
\begin{tabular}{ | l | l | l || l | l | l | }
    \hline 
Contract & \#methods & \#states & \#BFA & \#DFA & \#TOPL \\ 
 \hline 
CR-1 & 3 & 2 & 3 & 5 & 9 \\
 \hline 
CR-2 & 3 & 3 & 5 & 5 & 14 \\
 \hline 
CR-3 & 3 & 5 & 4  & 8 & 25 \\
 \hline 
CR-4 & 5 & 5 & 5 & 10 & 24 \\
 \hline 
CR-5 & 5 & 9 & 8 & 29 & 71 \\
 \hline 
CR-6 & 5 & 14 & 9 & 36 & 116 \\
\hline 
CR-7 & 7 & 18 & 12 & 85 & 213 \\
\hline 
CR-8 & 7 & 30 & 10 & 120 & 323 \\
\hline 
CR-9 & 7 & 41 & 12 & 157 & 460 \\
\hline 
\end{tabular}
& 
\begin{tabular}{ | l | l | l || l | l | l | }
    \hline 
Contract & \#methods & \#states & \#BFA & \#DFA & \#TOPL \\ 
 \hline 
CR-10 & 10 & 85 & 18 & 568 & 1407 \\
\hline 
CR-11 & 14 & 100 & 17 & 940 & 1884 \\
\hline 
CR-12 & 14 & 1044 & 32 & 7766 & 20704 \\
\hline 
CR-13 & 14 & 1628 & 21 & 13558 & 33740 \\
\hline 
CR-14 & 14 & 2322 & 21 & 15529 & 47068 \\
\hline 
CR-15 & 14 & 2644 & 24 & 26014 & 61846 \\
\hline 
CR-16 & 16 & 3138 & 29 & 38345 & 88134 \\
\hline 
CR-17 & 18 & 3638 & 23 & 39423 & 91120 \\
\hline 
CR-18 & 18 & 4000 & 27 & 41092 & 101185 \\
\hline 
\end{tabular}
\end{tabular} 

    }
\end{center}
\vspace*{-0.5cm}
    \caption{Details of the 18 contracts in our evaluation.}
    \label{fig:specsize}
    \vspace*{-0.0cm}
\end{figure}

\captionsetup[figure]{font=small}
\begin{figure*}[!t]
    \centering
    \begin{subfigure}[b]{\textwidth}
        \begin{subfigure}[b]{0.48\linewidth}
    \begin{tikzpicture}[scale=0.62]
        \begin{axis}[
            title={BFA vs DFA},
            xlabel={Number of states [k states]},
            ylabel={Time [in s]},
            xmin=0, xmax=4,
            ymin=0, ymax=6.5,
            xtick={0.1, 0.5, 1.0, 1.5, 2.0, 2.5, 3.0, 3.5, 4.0},
            ytick={0.5, 1.0,  
            1.5, 2.0, 2.5, 3.0, 3.5, 4.0, 4.5, 5.0, 5.5, 6.0, 6.5},
            legend pos=north west, 
            ymajorgrids=true,
            grid style=dashed,
            style=thick, 
        ]

       
  
   \addplot[
    color=blue,
    mark=o,
    ]
    coordinates {
      (0.10, 0.28)
      (0.52, 0.26)
      (1.04, 0.26)
      (1.63, 0.27)
      (2.32, 0.26)
      (2.64, 0.27)
      (3.14, 0.27)
      (3.64, 0.28)
      (4.00, 0.26)
    };
    \addlegendentry{BFA: 4}

   \addplot[
    color=red,
    mark=o,
    ]
    coordinates {
      (0.10, 0.36)
      (0.52, 0.74)
      (1.04, 1.06)
      (1.63, 1.48)
      (2.32, 1.75)
      (2.64, 1.78)
      (3.14, 2.32)
      (3.64, 2.29)
      (4.00, 2.29)
    };
    \addlegendentry{DFA: 4}

   \addplot[
    color=blue,
    mark=triangle,
    ]
    coordinates {
      (0.10, 0.31)
      (0.52, 0.33)
      (1.04, 0.31)
      (1.63, 0.34)
      (2.32, 0.31)
      (2.64, 0.31)
      (3.14, 0.33)
      (3.64, 0.32)
      (4.00, 0.31)
    };
    \addlegendentry{BFA: 6}

   \addplot[
    color=red,
    mark=triangle,
    ]
    coordinates {
      (0.10, 0.39)
      (0.52, 0.85)
      (1.04, 1.09)
      (1.63, 1.69)
      (2.32, 1.75)
      (2.64, 2.20)
      (3.14, 3.20)
      (3.64, 2.96)
      (4.00, 4.03)
    };
    \addlegendentry{DFA: 6}

   \addplot[
    color=blue,
    mark=diamond,
    ]
    coordinates {
      (0.10, 0.37)
      (0.52, 0.36)
      (1.04, 0.38)
      (1.63, 0.36)
      (2.32, 0.36)
      (2.64, 0.38)
      (3.14, 0.37)
      (3.64, 0.37)
      (4.00, 0.36)
    };
    \addlegendentry{BFA: 8}

   \addplot[
    color=red,
    mark=diamond,
    ]
    coordinates {
      (0.10, 0.49)
      (0.52, 1.03)
      (1.04, 1.41)
      (1.63, 2.09)
      (2.32, 2.05)
      (2.64, 3.40)
      (3.14, 3.90)
      (3.64, 3.35)
      (4.00, 4.48)
    };
    \addlegendentry{DFA: 8}

   \addplot[
    color=blue,
    mark=square,
    ]
    coordinates {
      (0.10, 0.43)
      (0.52, 0.42)
      (1.04, 0.44)
      (1.63, 0.42)
      (2.32, 0.42)
      (2.64, 0.42)
      (3.14, 0.43)
      (3.64, 0.41)
      (4.00, 0.40)
    };
    \addlegendentry{BFA: 10}

   \addplot[
    color=red,
    mark=square,
    ]
    coordinates {
      (0.10, 0.64)
      (0.52, 1.00)
      (1.04, 1.90)
      (1.63, 2.59)
      (2.32, 3.05)
      (2.64, 3.95)
      (3.14, 5.87)
      (3.64, 6.04)
      (4.00, 5.39)
    };
    \addlegendentry{DFA: 10}

        \end{axis}
        
        \end{tikzpicture}
            \vspace*{-0.5cm} 
            \caption{DFA vs \lfa execution comparison on composed contracts (500-1k LoC)\label{sf:dfa2}}
        \end{subfigure}\hfill
        \begin{subfigure}[b]{0.48\linewidth}
    \begin{tikzpicture}[scale=0.62]
        \begin{axis}[
            title={BFA vs DFA},
            xlabel={Number of states},
            ylabel={Time [in s]},
            xmin=0, xmax=90,
            ymin=0, ymax=11.0,
            xtick={5, 10.0, 20.0, 30.0, 40.0, 50.0, 60.0, 70.0, 80.0, 90},
            ytick={0.5, 1.0,  
            1.5, 2.0, 2.5, 3.0, 3.5, 4.0, 4.5, 5.0, 5.5, 6.0, 6.5, 
            7,7.5,8,8.5,9,9.5,10,10.5,11},
            legend pos=north west, 
            ymajorgrids=true,
            grid style=dashed,
            style=thick, 
        ]

       
  
   \addplot[
    color=blue,
    mark=o,
    ]
    coordinates {
      (5, 2.334)
      (9, 2.532)
      (14, 2.570)
      (18, 2.638)
      (30, 2.180)
      (41, 2.460)
      (85, 3.014)
      
    };
    \addlegendentry{BFA: 8}

   \addplot[
    color=red,
    mark=o,
    ]
    coordinates {
      (5, 2.380)
      (9, 2.610)
      (14, 2.554)
      (18, 3.730)
      (30, 3.428)
      (41, 3.392)
      (85, 5.216)
    };
    \addlegendentry{DFA: 8}

   \addplot[
    color=blue,
    mark=square,
    ]
    coordinates {
      (5, 2.900)
      (9, 3.198)
      (14, 2.828)
      (18, 3.452)
      (30, 2.542)
      (41, 3.068)
      (85, 3.484)    
    };
    \addlegendentry{BFA: 20}

   \addplot[
    color=red,
    mark=square,
    ]
    coordinates {
      (5, 3.160)
      (9, 3.866)
      (14, 3.554)
      (18, 6.644)
      (30, 4.568)
      (41, 8.968)
      (85, 10.260)   
    };
    \addlegendentry{DFA: 20}

        \end{axis}
        
        \end{tikzpicture}
              \vspace*{-0.5cm}
            \caption{DFA vs \lfa execution comparison on composed contracts (15k LoC)~\label{sf:dfa1}}
        \end{subfigure}\hfill 
    \end{subfigure}\hfill
    \begin{subfigure}[b]{\textwidth}
        \begin{subfigure}[b]{0.48\textwidth}

\begin{tikzpicture}[scale=0.62]
    \begin{axis}[
        title={BFA vs TOPL},
        xlabel={Number of states [k states]},
        ylabel={Time [in s]},
        xmin=0, xmax=4,
        ymin=0, ymax=450,
        ymode=log,
        xtick={0.1, 0.5, 1.0, 1.5, 2.0, 2.5, 3.0, 3.5, 4.0},
        legend pos=north west,
        ymajorgrids=true,
        grid style=dashed,
        style=thick, 
    ]

   

\addplot[
color=blue,
mark=o,
]
coordinates {
  (0.10, 0.28)
  (0.52, 0.26)
  (1.04, 0.26)
  (1.63, 0.27)
  (2.32, 0.26)
  (2.64, 0.27)
  (3.14, 0.27)
  (3.64, 0.28)
  (4.00, 0.26)
};
\addlegendentry{BFA: 4}

\addplot[
color=red,
mark=o,
]
coordinates {
  (0.10, 4.23)
  (0.52, 25.52)
  (1.04, 47.60)
  (1.63, 99.38)
  (2.32, 145.97)
  (2.64, 225.58)
  (3.14, 364.80)
  (3.64, 354.30)
  (4.00, 438.37)
};
\addlegendentry{TOPL: 4}

\addplot[
color=blue,
mark=triangle,
]
coordinates {
  (0.10, 0.31)
  (0.52, 0.33)
  (1.04, 0.31)
  (1.63, 0.34)
  (2.32, 0.31)
  (2.64, 0.31)
  (3.14, 0.33)
  (3.64, 0.32)
  (4.00, 0.31)
};
\addlegendentry{BFA: 6}

\addplot[
color=red,
mark=triangle,
]
coordinates {
    (0.10, 4.02)
    (0.52, 25.38)
    (1.04, 58.27)
    (1.63, 96.97)
    (2.32, 150.08)
    (2.64, 197.22)
    (3.14, 344.77)
    (3.64, 353.73)
    (4.00, 458.26)
};
\addlegendentry{TOPL: 6}

\addplot[
color=blue,
mark=diamond,
]
coordinates {
  (0.10, 0.37)
  (0.52, 0.36)
  (1.04, 0.38)
  (1.63, 0.36)
  (2.32, 0.36)
  (2.64, 0.38)
  (3.14, 0.37)
  (3.64, 0.37)
  (4.00, 0.36)
};
\addlegendentry{BFA: 8}

\addplot[
color=red,
mark=diamond,
]
coordinates {
  (0.10, 4.02)
  (0.52, 26.64)
  (1.04, 60.39)
  (1.63, 88.74)
  (2.32, 135.21)
  (2.64, 214.65)
  (3.14, 351.20)
  (3.64, 335.19)
  (4.00, 458.34)
};
\addlegendentry{TOPL: 8}

\addplot[
color=blue,
mark=square,
]
coordinates {
  (0.10, 0.43)
  (0.52, 0.42)
  (1.04, 0.44)
  (1.63, 0.42)
  (2.32, 0.42)
  (2.64, 0.42)
  (3.14, 0.43)
  (3.64, 0.41)
  (4.00, 0.40)
};
\addlegendentry{BFA: 10}

\addplot[
color=red,
mark=square,
]
coordinates {
  (0.10, 3.82)
  (0.52, 23.77)
  (1.04, 49.69)
  (1.63, 106.55)
  (2.32, 135.72)
  (2.64, 183.74)
  (3.14, 309.33)
  (3.64, 331.70)
  (4.00, 412.08)
};
\addlegendentry{TOPL: 10}

    \end{axis}
    
    \end{tikzpicture}
        \vspace*{-0.5cm}
        \caption{\textsc{Topl} vs \lfa comparison on composed contracts (500-1k LoC)\label{sf:topl2}}
        \end{subfigure}\hfill 
    \begin{subfigure}[b]{0.48\textwidth}

\begin{tikzpicture}[thick, scale=0.62]
    \begin{axis}[
        title={BFA vs TOPL},
        xlabel={Number of states},
        ylabel={Time [in s]},
        xmin=0, xmax=90,
        ymin=0, ymax=60,
        xtick={5, 10.0, 20.0, 30.0, 40.0, 50.0, 60.0, 70.0, 80.0, 90},
        ytick={3, 5, 10,20,30,40,50,60},
        legend pos=north west,
        ymajorgrids=true,
        grid style=dashed,
        style=thick, 
    ]

   

\addplot[
color=blue,
mark=o,
]
coordinates {
  (5, 2.260)
  (9, 2.470)
  (14, 2.590)
  (18, 2.690)
  (30, 2.160)
  (41, 2.450)
  (85, 3.090)  
};
\addlegendentry{BFA: 8}

\addplot[
color=red,
mark=o,
]
coordinates {
  (5, 12.830)
  (9, 10.740)
  (14, 12.420)
  (18, 16.120)
  (30, 12.040)
  (41, 16.220)
  (85, 51.420)  
};
\addlegendentry{TOPL: 8}

\addplot[
color=blue,
mark=square,
]
coordinates {
  (5, 2.880)
  (9, 3.110)
  (14, 2.770)
  (18, 3.490)
  (30, 2.600)
  (41, 2.980)
  (85, 3.560)
};
\addlegendentry{BFA: 20}

\addplot[
color=red,
mark=square,
]
coordinates {
  (5, 14.750)
  (9, 12.270)
  (14, 12.870)
  (18, 19.120)
  (30, 13.550)
  (41, 20.750)
  (85, 54.280)
};
\addlegendentry{TOPL: 20}

    \end{axis}
    
    \end{tikzpicture}
        \vspace*{-0.5cm}
        \caption{\textsc{Topl} vs \lfa comparison on composed contracts (15k LoC)\label{sf:topl1}}
    \end{subfigure}
\end{subfigure}
    \vspace*{-0.5cm}
    \caption{Runtime comparisons.\label{sf:exectimes}
        Each line represents a different number of
        base classes composed in a client code.
    }
\end{figure*}
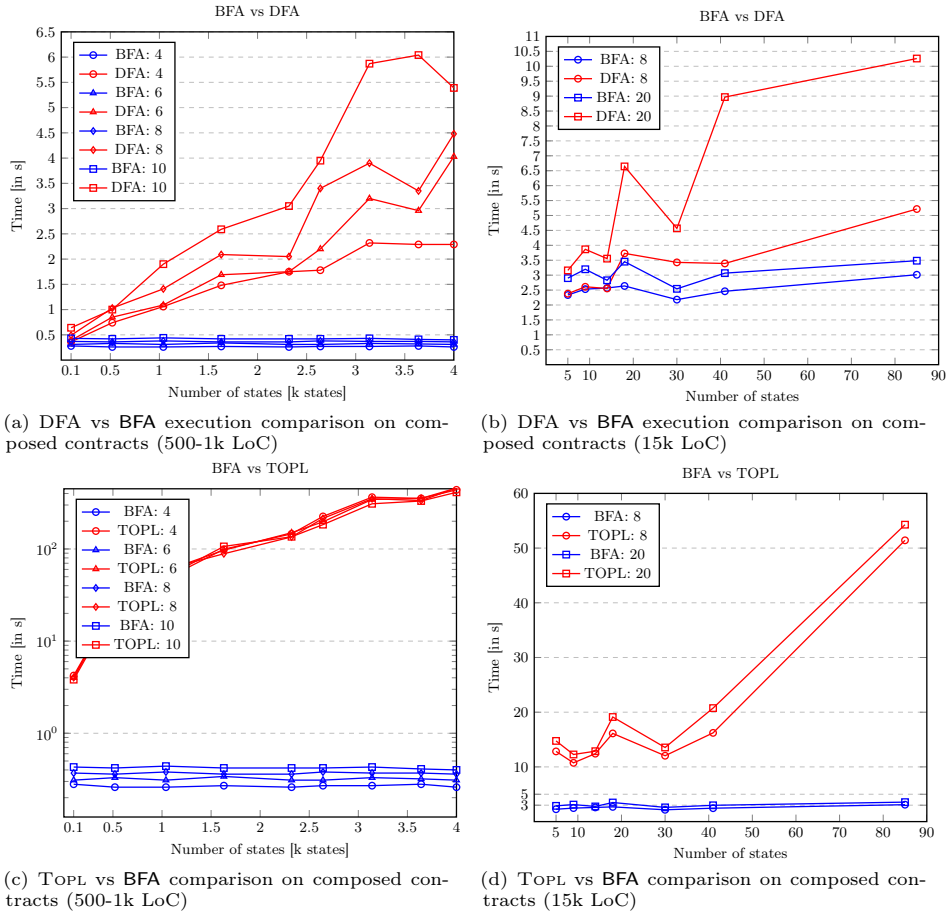

\paragraph{Performance Evaluation} 
\revision{Recall that we distinguish between \emph{base} and \emph{composed} classes: the former 
    have a user-entered contract, and the latter have contracts that are implicitly inferred based on those of
    their members (that could be either base or composed classes themselves).
    The total number of base classes in a composed class and contract size (i.e., the
    number of states in a minimal DFA that is a translation of a \lfa contract)
    play the most significant roles in execution-time. In~\figref{sf:exectimes} we present a
comparison of analyzer execution-times (y-axis) with contract size (x-axis),
where each line in the graph represents a different number of base classes composed
in a given class (given in legends).} 

\noindent\emph{Comparing \lfa analysis against  DFA
analysis.~} 
\textbf{Fig.~\ref{sf:dfa2}} compares various class compositions (with contracts)
specified in the legend, for client programs of 500-1K LoC. The DFA
implementation sharply increases in execution-time as the number of states
increases. The \lfa implementation remains rather constant, always under the SLA
of 1 seconds. Overall, \lfa produces a geometric mean speedup over DFA of
5.52$\times$. \textbf{Fig.~\ref{sf:dfa1}} compares various class compositions
for client programs of 15K LoC. Both implementations fail to meet the SLA;
however, the \lfa is close and exhibits constant behaviour regardless of the
number of states in the contract. The DFA implementation is rather erratic,
tending to sharply increase in execution-time as the number of states increases.
Overall, \lfa produces a geometric mean speedup over DFA of 1.5$\times$.

\noindent\emph{Comparing \lfa-based analysis vs TOPL typestate 
implementations (Execution time).~} 
Here again client programs do not require aliasing. \textbf{Fig.~\ref{sf:topl2}}
compares various class compositions for client programs of 500-1K LoC. The \textsc{Topl}
implementation sharply increases in execution-time as the number of states
increases, quickly missing the SLA. In contrast, the \lfa implementation remains
constant always under the SLA. Overall, \lfa produces a geometric mean speedup
over \textsc{Topl} of 6.59$\times$. \textbf{Fig.~\ref{sf:topl1}} compares various class
compositions for client programs of 15K LoC. Both implementations fail to meet
the SLA. The \textsc{Topl} implementation remains constant until $\sim$30 states and then
rapidly increases in execution time. Overall, \lfa produces a geometric mean
speedup over \textsc{Topl} of 287.86$\times$. 

Overall, we validate Claim-II by showing that our technique removes state as a
factor of performance degradation at the expense of limited but suffice contract
expressively. Even when using client programs of 15K LoC, we remain close to our
SLA and with potential to achieve it with further optimizations.



\section{Related Work}
\label{sec:relatedwork}
We focus on comparisons with restricted forms of typestate contracts. We refer
to the typestate
literature~\cite{DBLP:journals/tse/StromY86,lam,DeLineECOOP2004,clara,fugue} for
a more general treatment. The work~\cite{KelloggRSSE2020} proposes restricted
form of typestates tailored for use-case of the object construction using the
builder pattern. This approach is restricted in that it only accumulates called
methods in an abstract (monotonic) state, and it does not require aliasing for
supported contracts. Compared to our approach, we share the idea of specifying
typestate without explicitly mentioning states. On the other hand, their
technique is less expressive than our annotations. They cannot express various
properties we can (e.g., the property ``cannot call a method''). Similarly,
\cite{fahndrich2003heap} defines heap-monotonic typestates where  monotonicity
can be seen as a restriction. It can be performed without an alias analysis. 

Recent work on the \textsc{Rapid} analyzer~\cite{RAPID} aims to verify
cloud-based APIs usage.  It combines \emph{local} type-state with global
value-flow analysis. Locality of type-state checking in their work is related to
aliasing, not to type-state specification as  in our work. Their type-state
approach is DFA-based. They also highlight the state explosion problem for usual
contracts found in practice, where the set of methods has to be invoked prior to
some event. In comparison, we allow more granular contract specifications with a
very large number of states while avoiding an explicit DFA. The \textsc{Fugue}
tool~\cite{fugue} allows DFA-based specifications, but also annotations for
describing specific \emph{resource protocols} contracts. These annotations have
a \emph{locality} flavor---annotations on one method do not refer to other
methods. Moreover, we share the idea of specifying typestate without explicitly
mentioning states. These additional annotations in \textsc{Fugue}  are more
expressive than DFA-based typestates (e.g. ``must call a release method''). 
\revision{We conjecture that ``must call''  property can be encoded as 
bit-vectors in a complementary way to our \lfa approach.   
We leave this extension for future work. }

Our annotations could be mimicked by having a local DFA attached to each method.
In this case, the DFAs would have the same restrictions as our annotation
language. We are not aware of prior work in this direction. We also note that
while our technique is implemented in \textsc{Infer} using the  algorithm in
\S\ref{sec:technical}, the fact that we can translate typestates to bit-vectors
allows typestate analysis for local contracts to be used in distributive
dataflow frameworks, such as IFDS~\cite{IFDS}, without the need for modifying
the framework for non-distributive domains~\cite{extIFDS}.

\vspace{-3mm}
\section{Concluding Remarks}
\label{sec:conclusion}
In this paper, we have tackled the problem of analyzing code contracts in
 low-latency environments by developing a novel lightweight typestate analysis.
 Our technique is based on \lfas, a sub-class of contracts that can be encoded
 as bit-vectors. We believe \lfas are a simple and effective abstraction, with
 substantial potential to be ported to other settings in which DFAs are normally
 used.

 \paragraph{Acknowledgements}
  We are grateful to the anonymous reviewers for their constructive remarks. This 
work has been partially supported by the Dutch Research Council (NWO) under
 project No. 016.Vidi.189.046 (Unifying Correctness for Communicating Software).

\bibliographystyle{splncs04}
\bibliography{lfa}

\newpage
\setcounter{tocdepth}{3}

\longversion{\appendix 

\newpage
\section{Proofs}

\thmbfa* 
\label{app:tbfa}

\begin{proof} We only consider the first item, as the second item is shown similarly. 
   By $\wtd p \cdot m_{n+1} \notin L$ and 
   $\wtd p \cdot m_n \cdot m_{n+1} \in L$ and \defref{d:lfa} we know that 
   \begin{align}
    m_{n+1} \in E_{n}
    \label{p:bfa-1}
   \end{align}
   Further, for any $\wtd m \in \Sigma^*_c$ let $q_{b}$ be such that  
    $\delta(q_{10^{n-1}}, \wtd m)=q_b$   
   and $q_{b'}$ s.t. $\delta(q_b, m_n)=q_{b'}$.  Now, by the definition of 
   \defref{d:lfa} we have that $\delta(q_{b'},m_{n+1})$ 
   is defined as  by \eqref{p:bfa-1} we know 
   $P_{n+1} = \{ m_{n+1} \} \subseteq b'$. Thus, for all $\wtd m \in L$
   we have
   $\wtd m \cdot m_n \cdot m_{n+1} \in L$. This concludes the proof. 

 \end{proof} 

\thmstates* 
\label{app:states}

\begin{proof}  
	We show two items: 
	\begin{enumerate} 
		\item By \defref{d:lfa}, for all $q_b \in P$ we know  $\delta(q_b, m)$ is defined  
		when $P \subseteq b$ with $\langle E, P, D \rangle = \mathcal{L}_c(m)$. 
		So, we have $P \subseteq \bigcap_{q_b \in P} b = b_*$ and 
		$\delta(q_{b_*}, m)$ is defined. 
		\item  
	By induction on $\len{P}$. 
	\begin{itemize}
		\item $\len{P} = 1$. Follows immediately as $\bigcap_{q_b \in \{q_b\}} q_b = q_b$. 
		\item $\len{P} > 1$. Let $P = P_0 \cup \{q_b\}$. 
		Let $\len{P_0}=n$. By IH we know  
		\begin{align} 
			\label{eq:pr-inter-ih1}
		 \bigcap_{q_b \in P_0} \map{\sigma}(b) = \map{\sigma}(\bigcap_{q_b \in P_0}b)
		\end{align}
		\noindent We should show 
		$$\bigcap_{q_b \in (P_0 \cup \{q_{b'}\})} \map{\sigma}(b) = 
		\map{\sigma}(\bigcap_{q_b \in (P_0 \cup \{q_{b'}\})}b)$$
		We have 
		\begin{align*} 
			\bigcap_{q_b \in (P_0 \cup \{q_{b'}\})} \map{\sigma}(b) &=  
			\bigcap_{q_b \in P_0} \map{\sigma}(b) \cap \map{\sigma}(b') & \\ 
			&= \map{\sigma}(b_{*}) \cap \map{\sigma}(b') & (\text{by \eqref{eq:pr-inter-ih1}})\\ 
			&= ((b_* \cup E) \setminus D) \cap ((b' \cup E) \setminus D) &  \\ 
			& = ((b_* \cap b') \cup E) \setminus D & (\text{by set laws}) \\
			&= \map{\sigma}(b_*\ \cap \ b') = \map{\sigma}(\bigcap_{q_b \in (P_0 \cup \{q_{b'}\})}b) & 
		\end{align*} 
		\noindent where $b_*= \map{\sigma}(\bigcap_{q_b \in P_0}b)$.  
		This concludes the proof. 

	\end{itemize}
\end{enumerate} 
\end{proof} 

\thmjoin* 
\label{app:join}

\begin{proof} 
	By set laws we have: 
	\begin{align*}
		\map{\phi_1}(b) \cap \map{\phi_2}(b) &= ((b \cup E_1) \setminus D_1) \cap ((b \cup E_2) \setminus D_2)  \\
		&= ((b \cup E_1) \cap (b \cup E_2)) \setminus (D_1 \cup D_2) \\ 
		&= (b \cup (E_1 \cap E_2)) \setminus (D_1 \cup D_2) \\
		&= (b \cup (E_1 \cap E_2 \setminus (D_1 \cup D_2)) \setminus (D_1 \cup D_2) 
		= \map{\phi_1 \sqcup \phi_2}(b)
	\end{align*}
	This concludes the proof. 
\end{proof}

\thmdtranfer*
\label{app:transfer}

\begin{proof} 
	\begin{itemize} 
		\item ($\Rightarrow$) Soundness: 
		By induction on the length of method sequence $\widetilde{m} = m_1,\ldots,m_n$. 
		\begin{itemize} 
			\item Case $n = 1$. In this case we have $\widetilde{m} = m_1$. 
				Let $\langle E^m, D^m, \{m_1\} \rangle = \mathcal{L}_c(m_1)$. 
				By \defref{d:transfer} we have 
				$E' = (\emptyset \cup E^m) \setminus D^m = E^m$ and 
				$D' = (\emptyset \cup D^m) \setminus E^m = D^m$  as 
				$E^m$ and $D^m$ are disjoint, and 
				$P' = \emptyset \cup (\{m \} \setminus \emptyset)$.  
				So, we have $b' = (b \cup E^m) \setminus D^m$. 
				Further, we have
				 $P' \subseteq b$. 
				Finally, by the definition of $\delta(\cdot)$ from 
				\defref{d:lfa} we have 
				$\hat{\delta}(q_b, m_1,\ldots,m_n)=q_{b'}$. 

				\item Case $n > 1$. Let $\widetilde{m}=m_1,\ldots,m_n, m_{n+1}$. 
		By IH we know  
		\begin{align}
		\label{eq:p1-ih} 
			&\dtransfer_{c}(m_1,\ldots,m_n, \langle \emptyset, 
			\emptyset, \emptyset \rangle) =
				\langle E', D', P' \rangle
				\Rightarrow \hat{\delta}(q_b, m_1,\ldots,m_n)=q_b'
		\end{align}
		\noindent where 
		$b' = (b \cup E') \setminus D'$ and  $P' \subseteq b$. 
		Now, we assume $P'' \subseteq b$ and 
		\begin{align*}
			&\dtransfer_{c}(m_1,\ldots,m_n, m_{n+1}, 
			\langle \emptyset, \emptyset, \emptyset \rangle) =
			\langle E'', D'', P'' \rangle
		\end{align*}
		We should show  
		\begin{align} 
			\label{eq:p1-goal} 
			\hat{\delta}(q_b, m_1,\ldots,m_n, m_{n+1})=q_b'' 
		\end{align} 
		\noindent where
		$b'' = (b \cup E'') \setminus D''$. 
		Let $\mathcal{L}_c(m_{n+1})=\langle E^m, D^m, P^m \rangle$.  
		We know $P^m=\{m_{n+1}\}$. 
		By \defref{d:transfer} we have 
		\begin{align*}
			\dtransfer_{c}(m_1,\ldots,m_n, m_{n+1}, 
			\langle \emptyset, \emptyset, \emptyset \rangle) = 
			\dtransfer_{c}(m_{n+1}, \langle E', D', P' \rangle)
		\end{align*}

		Further, we have 
		\begin{align} 
			\label{eq:p1-epd-ih}
			E'' = (E' \cup E^m) \setminus D^m \qquad  
			D'' = (D' \cup D^m) \setminus E^m \qquad   
			P'' = P' \cup (P^m \setminus E') 
		\end{align} 

		Now, by substitution and De Morgan's laws we have:
		\begin{align*}
			b'' &= (b \cup E'') \setminus D'' = \\ 
			&=
			(b \cup ((E' \cup E^m) \setminus \textcolor{black}{D^m})) 
			\setminus ((D' \cup \textcolor{black}{D^m}) \setminus E^m)\\
			&= 
			((b \cup (E' \cup E^m)) 
			\setminus (D' \setminus E^m)) \setminus \textcolor{black}{D^m}  \\
			&= 
			(((b \cup E') \setminus D') \cup \textcolor{black}{E^m}) \setminus D^m \\
			&= (b' \cup \textcolor{black}{E^m}) \setminus D^m
		\end{align*}

		Further, by $P'' \subseteq b$,  
		$P'' = P' \cup (P^m \setminus E')$, and $P^m \cap D'=\emptyset$, we have 
		$P^m \subseteq (b \cup E') \setminus D' = b'$ (by \eqref{eq:p1-ih}).
		So, we can see that by definition of \defref{d:lfa} we have 
		$\delta(q_{b'}, m_{n+1})=q_{b''}$. 
		This concludes this case. 
	\end{itemize} 

		\item ($\Leftarrow$) Completeness: 
		\begin{itemize} 
			\item $n = 1$. In this case $\widetilde{m}=m_1$. 
			Let $\langle E^m, D^m, \{m_1\} \rangle = \mathcal{L}_c(m_1)$. 
			By \defref{d:lfa} we have $b'=(b \cup E^m) \setminus D^m$ 
			and $\{m_1\} \subseteq b$. By \defref{d:transfer} we have 
			$E' = E^m$, $D'=D^m$, and $P' = \{m_1\} $. 
			Thus, as $\{m_1\} \cap \emptyset = \emptyset$ we have 
			$b' = \map{\langle E', D', P' \rangle}(b)$. 
			
		\item $n > 1$. Let $\widetilde{m}=m_1,\ldots,m_n, m_{n+1}$. 
		By IH we know  
		\begin{align}
		\label{eq:p1-ih2} 
			&
			\hat{\delta}(q_b, m_1,\ldots,m_n)=q_b'
			\Rightarrow  
			\dtransfer_{c}(m_1,\ldots,m_n, \langle \emptyset, 
			\emptyset, \emptyset \rangle) =
				\langle E', D', P' \rangle
		\end{align}
		\noindent where $b' = (b \cup E') \setminus D'$ and  $P' \subseteq b$. 
		Now, we assume 
		\begin{align}
			\label{eq:p1-hyp2}
			\hat{\delta}(q_b, m_1,\ldots,m_n, m_{n+1})=q_{b''}
		\end{align}
		We should show that 
		\begin{align*}
			&\dtransfer_{c}(m_1,\ldots,m_n, m_{n+1}, 
			\langle \emptyset, \emptyset, \emptyset \rangle) =
			\langle E'', D'', P'' \rangle
		\end{align*}
		\noindent such that $b'' = (b \cup E'') \setminus D''$ and 
		$P'' \subseteq b$. 
		We know 
		\begin{align*}
			\dtransfer_{c}(m_1,\ldots,m_n, m_{n+1}, 
			\langle \emptyset, \emptyset, \emptyset \rangle) = 
			\dtransfer_{c}(m_{n+1}, \langle E', D', P' \rangle)
		\end{align*}

		By \defref{d:lfa} we have: 
		\begin{align*} 
			\hat{\delta}(q_b, m_1,\ldots,m_n, m_{n+1})= 
			\delta( \hat{\delta}(q_b, m_1,\ldots,m_n), m_{n+1})=q_{b''}
		\end{align*}

		So by \eqref{eq:p1-ih2} and 
		\eqref{eq:p1-hyp2} we have $\{m_{n+1} \} \subseteq b'$ 
		and $b'=(b \cup E') \setminus D'$. 
		It follows $\{m_{n+1} \} \cap D' = \emptyset$. 
		That is, $\dtransfer_{c}(m_{n+1}, \langle E', D', P' \rangle)$
		is defined. Finally, showing that 
		$b'' = (b \cup E'') \setminus D''$  is by the substitution and 
		De Morgan's laws as in the previous case. 
		This concludes the proof. 
	\end{itemize} 
 	\end{itemize} 
	 Now, we discuss specialization of \thmref{t:dtranfer} for the code checking. 
	In this case, we know that a method sequence starts with the
	constructor method (i.e., the sequence is of the form $constr, m_1, \ldots,
	m_n$) and $q_{10^{n-1}}$ is the input state. By  
	$\mathit{well\_formed}(\mathcal{L}_c)$
	(\defref{d:wellformed-lc}) we know that if $\delta(q_{10^{n-1}},
	constr)=q_b$ and 
	$$\dtransfer_c(constr, m_1, \ldots, m_n, \langle \emptyset,
	\emptyset, \emptyset \rangle) = \sigma$$ then methods not
	enabled in $q_b$ are in the disabling set of 
	$\sigma$. Thus, for any sequence $m_1,\ldots, m_{k-1}, m_{k}$ such that 
	$m_k$ is disabled by the constructor  and 
	not enabled in substring $m_1,\ldots,
	m_{k-1}$, the condition $P \cap D_i \not= \emptyset$ correctly  
	 checks that a method is disabled. If  $\mathit{well\_formed}(\mathcal{L}_c)$
	 did not hold, the algorithm would fail to detect an error as it would put
	 $m_k$ in $P$ since $m_k \notin E$. 
\end{proof}


\newpage
\section{Sample Contract used in Evaluations (\secref{sec:evaluation})}
\label{app:cr4}
	\begin{lstlisting}[caption={SparseLU LFA CR4 contract}, label={lst:refltaext}, captionpos=b]
class SparseLU {
	SparseLU(); 
	$@EnableOnly(factorize)$
	void analyzePattern(Mat a); 
	$@EnableOnly(solve, transpose)$
	void factorize(Mat a); 
	$@EnableOnly(solve, transpose)$
	void compute(Mat a); 
	$@EnableAll$
	void solve(Mat b); 
	$@Disable(transpose)$
	void transpose(); }
\end{lstlisting}

\begin{lstlisting}[caption={SparseLU DFA CR4 contract}, label={lst:refdfatranspose}, captionpos=b]
class SparseLU {
	$states q0, q1, q2, q3, q4;$ 
	$@Pre(q0) @Post(q1)$
	$@Pre(q3) @Post(q1)$ 
	void analyzePattern(Mat a); 
	$@Pre(q1) @Post(q2)$
	$@Pre(q3) @Post(q2)$  
	void factorize(Mat a); 
	$@Pre(q0) @Post(q2)$ 
	$@Pre(q3) @Post(q2)$ 
	void compute(Mat a); 
	$@Pre(q2) @Post(q3)$
	$@Pre(q3)$ 
	void solve(Mat b);  
	$@Pre(q2) @Post(q4)$
	$@Pre(q4) @Post(q3)$
	void transpose();}
\end{lstlisting}

\begin{figure}[h] 
    \begin{mdframed}
        \center
    \begin{tikzpicture}[shorten >=1pt,node distance=3.5cm,on grid,auto] 
       \node[state,initial] (q_0)   {$q_0$}; 
       \node[state](q_1) [right=of q_0] {$q_1$};
       \node[state](q_2) [below=of q_1] {$q_2$};
       \node[state](q_3) [right=of q_1] {$q_3$};
       \node[state](q_4) [right=of q_2] {$q_4$}; 
       \path[->]
       (q_0) edge [above] node {$aP$} (q_1)
       (q_1) edge  [left] node {$factorize$} (q_2)
       (q_0) edge [bend right, below, sloped] node {$compute$} (q_2)
       (q_2) edge [bend left=15] node {$solve$} (q_3)
       (q_3) edge [bend left=15, below, sloped] node {$comp., fact.$} (q_2)
       (q_3) edge [right, above] node {$aP$} (q_1)
       (q_3) edge [loop right] node{$solve$} (q_3) 
       (q_3) edge [loop right] node{$solve$} (q_3) 
       (q_2) edge [below] node{$transpose$} (q_4) 
       (q_4) edge [bend left=15, above] node {$solve$} (q_3)
       (q_3) edge [bend left=15, below] node {$transpose$} (q_4)
       ;
\end{tikzpicture}
\end{mdframed}
\caption{DFA diagram of SparseLU CR-4 contract}
\end{figure}
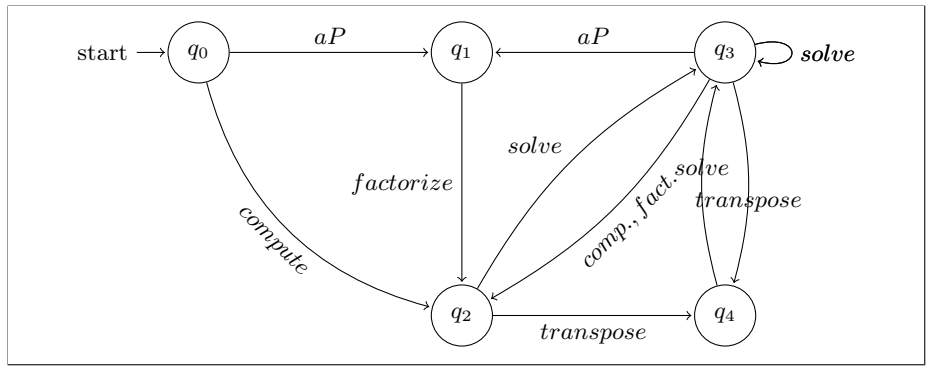 

\begin{lstlisting}[caption={SparseLU TOPL CR4 contract}, label={lst:reftopltranspose}, captionpos=b]
property SparseLU
  prefix "SparseLU"
  start -> start: *
  start -> q0: SparseLU() => x := RetFoo
  q1 -> q2: analyzePattern(SparseLU, IgnoreRet) when SparseLU == x
  q3 -> q2: analyzePattern(SparseLU, IgnoreRet) when SparseLU == x
  q1 -> q2: factorize(SparseLU, IgnoreRet) when SparseLU == x
  q3 -> q2: factorize(SparseLU, IgnoreRet) when SparseLU == x
  q0 -> q2: compute(SparseLU, IgnoreRet) when SparseLU == x
  q3 -> q2: compute(SparseLU, IgnoreRet) when SparseLU == x
  q2 -> q3: solve(SparseLU, IgnoreRet) when SparseLU == x
  q2 -> q4: transpose(SparseLU, IgnoreRet) when SparseLU == x
  q4 -> q2: transpose(SparseLU, IgnoreRet) when SparseLU == x
  q2 -> error: analyzePattern(SparseLU, IgnoreRet) when SparseLU == x
  q3 -> error: analyzePattern(SparseLU, IgnoreRet) when SparseLU == x
  q4 -> error: analyzePattern(SparseLU, IgnoreRet) when SparseLU == x
  q0 -> error: factorize(SparseLU, IgnoreRet) when SparseLU == x
  q2 -> error: factorize(SparseLU, IgnoreRet) when SparseLU == x
  q4 -> error: factorize(SparseLU, IgnoreRet) when SparseLU == x
  q1 -> error: compute(SparseLU, IgnoreRet) when SparseLU == x
  q2 -> error: compute(SparseLU, IgnoreRet) when SparseLU == x
  q4 -> error: compute(SparseLU, IgnoreRet) when SparseLU == x
  q1 -> error: solve(SparseLU, IgnoreRet) when SparseLU == x
  q4 -> error: solve(SparseLU, IgnoreRet) when SparseLU == x
  q4 -> error: solve(SparseLU, IgnoreRet) when SparseLU == x
  q0 -> error: transpose(SparseLU, IgnoreRet) when SparseLU == x
  q1 -> error: transpose(SparseLU, IgnoreRet) when SparseLU == x
  q4 -> error: transpose(SparseLU, IgnoreRet) when SparseLU == x
\end{lstlisting}}{}

\end{document}
\endinput